\documentclass[11pt]{article}
\usepackage{amsfonts}
\usepackage{amsmath}
\usepackage{amssymb}
\usepackage{enumerate}
\usepackage{natbib}
\usepackage{color}
\usepackage{array}
\usepackage{graphicx,epstopdf}
\usepackage[margin=1in]{geometry}
 \usepackage{float}
 \usepackage{comment}
\usepackage{rotating}
\usepackage{url}
\usepackage{subfigure}
 \usepackage[margin=1.5cm]{caption}
\usepackage{setspace}
\usepackage{hyperref}
\newtheorem{proposition}{Proposition}[section]
\newtheorem{example}{Example}[section]
\newtheorem{remark}{Remark}[section]

\newenvironment{proof}[1][Proof]{\text{#1.} }{\ \rule{0.5em}{0.5em}}
\def\Q{{\mathbb Q}}        % rationals
\def\R{{\mathbb R}}        % reals
\def\P{{\mathbb P}}        % probability
\def\E{{\mathbb E}}        % expectation
\def\1{{\mathbf 1}}        % indicator
\def\F{{\mathcal F}}        % potential measure
\DeclareMathOperator{\Tr}{Tr}
\numberwithin{theorem}{subsection}
\numberwithin{equation}{section}

\newcommand{\white}[1]{\textcolor{white}{#1}}

\begin{document}
\title{Dynamic Index Tracking and Risk Exposure Control Using Derivatives\thanks{The authors acknowledge the support from  KCG Holdings Inc., and the helpful remarks from the participants of the 2015 INFORMS Annual Meeting,  2016 SIAM Conference on Financial Math \& Engineering, Thalesians Quantitative Finance Seminar, and Quant Summit USA  2016.}}
\author{Tim Leung\thanks{Applied Mathematics Department,  University of Washington, Seattle WA 98195. Email: {timleung@uw.edu}. {Corresponding author}. } \and Brian Ward\thanks{Industrial Engineering \& Operations Research (IEOR) Department, Columbia University, New York, NY 10027. Email: {bmw2150@columbia.edu}.} }

\maketitle

\abstract{We develop a methodology  for index tracking and risk exposure control using financial derivatives. Under a continuous-time diffusion framework for price evolution, we present a pathwise approach to construct dynamic portfolios of derivatives in order to gain exposure to an index and/or market factors that may be not directly  tradable. Among our results, we establish a general tracking condition that  relates the portfolio drift to the desired exposure coefficients under any given model.  We also derive a slippage process that reveals how the portfolio return deviates from the targeted return. In our multi-factor setting, the portfolio's realized slippage depends not only on the realized variance of the index, but also the realized covariance among the index and factors. We implement our  trading strategies under a number of models, and compare the tracking strategies and performances when using different derivatives, such as futures and options. }
\medskip



  
\newpage

\section{Introduction}A common challenge faced by many institutional and  retail investors is to effectively control risk exposure to various market factors.  There is a great variety of indices designed to provide  different types of exposures across  sectors and asset classes, including equities, fixed income, commodities, currencies, credit risk, and more. Some of these indices  can be difficult or impossible to trade directly, but investors can trade the associated financial derivatives if they are available in the market.  For example, the CBOE Volatility Index (VIX),  often referred to as the  {fear index},  is  not directly tradable, but investors can gain exposure to the index and potentially hedge against market turmoil  by trading futures and options written on VIX.

To illustrate the benefits of having exposure to VIX, consider the 2011 U.S. credit rating downgrade by Standard and Poor's. News of a negative outlook by S\&P of the U.S. credit rating broke on April 18th, 2011.\footnote{See New York Times article: \url{http://www.nytimes.com/2011/04/19/business/19markets.html}.} As displayed in  Figure \ref{fig:VIX_SPY}(a), a  portfolio   holding only the SPDR S\&P 500 ETF (SPY) would go on to lose   about 10\%  with a volatile trajectory for a few months past the  official downgrade on August 5th, 2011.\footnote{See: \url{http://www.nytimes.com/2011/08/06/business/us-debt-downgraded-by-sp.html}} In contrast, a hypothetical portfolio with a mix of SPY (90\%) and VIX (10\%)  would  be stable through the downgrade and end up with a positive return. Figure \ref{fig:VIX_SPY}(b) shows the same pair of portfolios over the year 2014. Both earned roughly the same 15\% return though SPY  alone was  visibly more volatile than the  portfolio with SPY and VIX. The large drawdowns (for example on October 15th, 2014) were met by rises in VIX, creating a  stabilizing effect on the portfolio's value.

\begin{figure}[h]
	\begin{centering}
		\subfigure[April 2011 - Sept. 2011]{\includegraphics[trim={1.2cm 0.7cm 1.2cm 0.5cm},clip,width=3in]{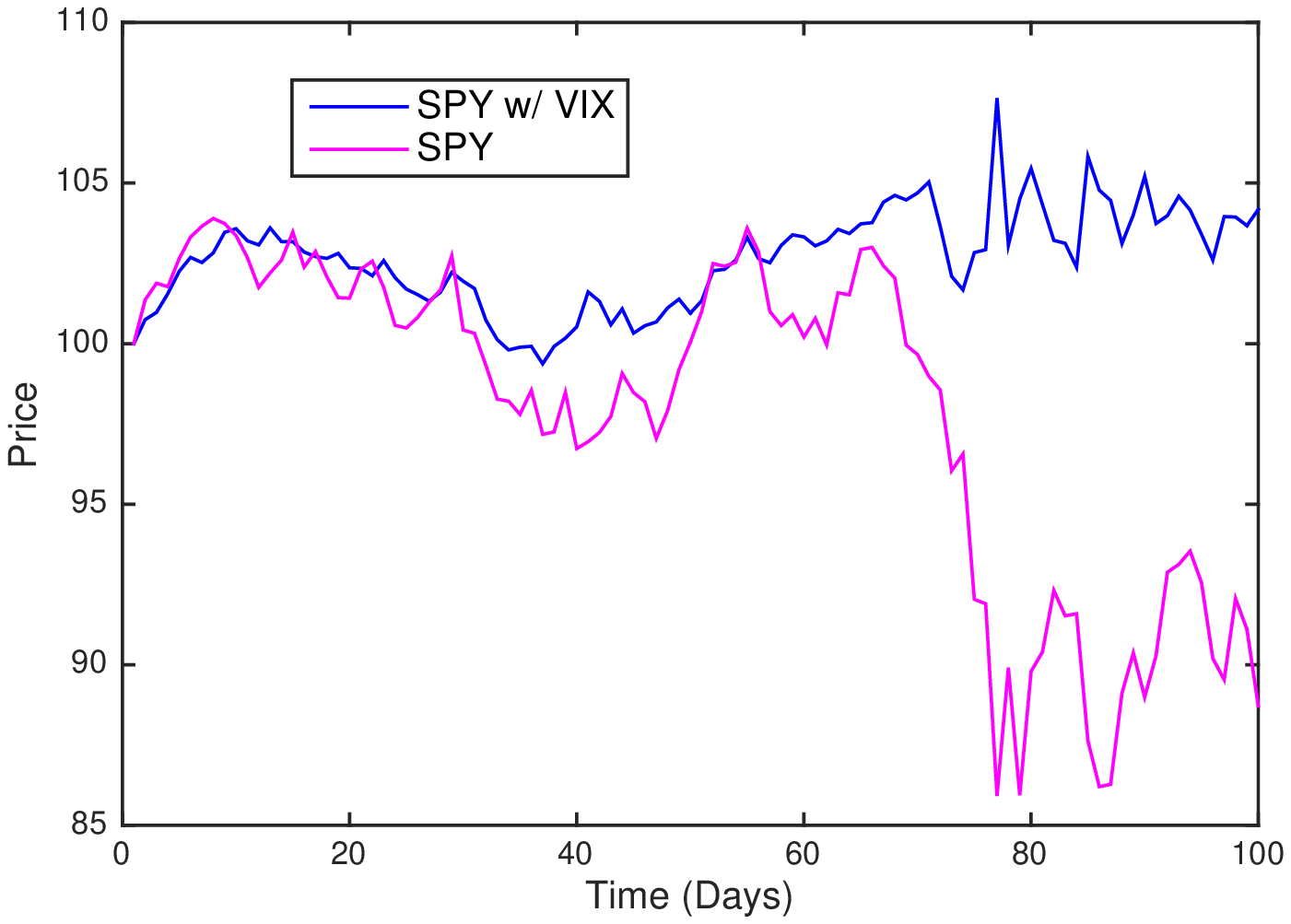}}
		\subfigure[Jan. 2014 - Dec. 2014]{\includegraphics[trim={1.2cm 0.7cm 1.2cm 0.5cm},clip,width=3in]{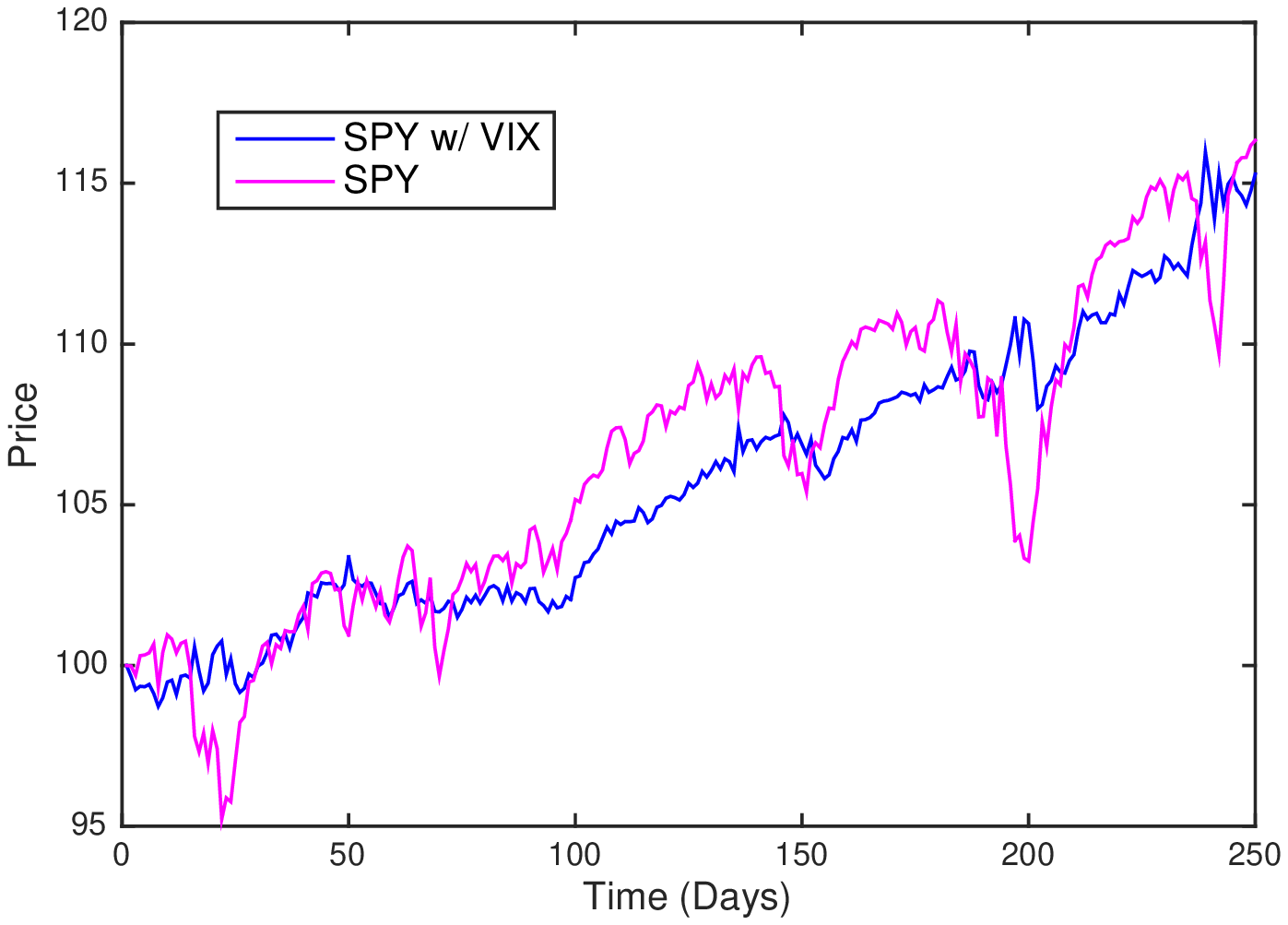}}
		\caption{\small{Hedged portfolios of SPY with 10\%   of wealth invested in VIX during  (a) April 2011 -- Sept. 2011, and (b) Jan. 2014 -- Dec. 2014.}}
		\label{fig:VIX_SPY}
	\end{centering}
\end{figure}

This example motivates the  investigation of    trading strategies that directly track VIX and other indices, or achieve  any pre-specified exposure with respect to an index or market factor.  Many ETFs or ETNs are advertised to provide exposure to an index by maintaining a portfolio of securities   related to the index, such as futures, options, and swaps.  However, some of ETFs or ETNs often miss their stated targets, and some tend to significantly underperform over time relative to the targets. One example is the Barclay's iPath S\&P 500 VIX Short-Term Futures ETN (VXX), which is also the most popular VIX exchange-traded product.\footnote{As measured by an average daily volume in excess of $57$ million shares as of May, 2017. \url{http://etfdb.com/etfdb-category/volatility/}} The failure of VXX to track VIX is well documented (e.g. \cite{HussonMcCann},  \cite{DengMcCannWang},  \cite{alexander2013volatility}, and \cite{whaleyVolCost}). In fact, most of these  ETFs or ETNs follow a static strategy or time-deterministic allocation that does not adapt to the rapidly  changing market. The problem of tracking and risk exposure control is relevant in all asset classes, and the use of derivatives is also common in other markets.  For example, many investors seek exposure to gold  to hedge against market turmoil. However, direct investment in gold bullion is difficult due to storage cost. In order to gain   exposure, an investor may select among  a number of gold derivatives, and ETFs  as well as their leveraged counterparts. An analysis on  the tracking properties of portfolios of gold futures and ETFs can be found in  \cite{LeungWard}. For an empirical study on  the tracking errors of a large collection of commodity leveraged and non-leveraged ETFs, we refer to  \cite{guoleung} .

In this paper, we discuss a general methodology  for index tracking and risk exposure control using derivatives.    In Section \ref{sec:cts_time_framework}, we describe the market in  a continuous-time diffusion framework,  and derive a condition (see Proposition \ref{prop:rep_cond}) that links the exposures attainable by a derivatives portfolio. In the special case that the exposures are constant over time, the portfolio value admits an explicit expression in terms of the reference index (see Proposition \ref{prop:pf_dynamics}).  In particular, we   quantify the divergence of portfolio return  from the target returns  of the index and its factors via the \emph{slippage process}.  With exposure to the index and multiple stochastic factors,  the slippage process  is a function of  not only the  realized variance of the underlying factors, but also  the  \emph{realized covariance} among  the index and   factors. The slippage process derived herein reveals the potential effect of portfolio value erosion arising from the interactions among all sources of randomness in any given model. Moreover,  it  can also explain as a special case the well-known volatility decay effect in leverage ETFs (see, among others, \cite{AZ}, \cite{jarrowLETFs},   and   \cite{LeungSantoliBook}). 

Index tracking or risk exposure control  can be perceived as an inverse problem to   dynamic hedging of derivatives. In the traditional hedging problem, the goal is to trade the underlying assets so as to replicate the price evolution of the derivative in question, and thus, the tradability of the underlying is of crucial importance. In our proposed paradigm, the index and   stochastic factors may not be directly traded, but there exist traded derivatives written on them. We use derivatives to track or, more generally, control risk exposure with respect to   the index return in a pathwise manner. Consequently, we can study the path properties resulting from  various portfolios of derivatives, and quantify the portfolio's divergence, if any,   from a pre-specified  benchmark.  Our methodology also allows the investor to achieve leveraged or non-leveraged exposures with respect to the associated factors in the model.

 The problem of index tracking has been well studied   from other different perspectives. Utilizing  a small subset of available stocks under the geometric Brownian Motion (GBM) model, \cite{YaoZhangZhou} solve a stochastic linear quadratic control problem to determine the optimal asset allocation  that best tracks a benchmark's return.  \cite{PrimbsSung} consider a variation of this problem by including  probabilistic portfolio constraints, e.g. short sale restrictions, and find  the optimal allocation via semi-definite programming. A similar problem is studied by \cite{MVOindexTracking}, who optimizes a combination of mean and variance of the portfolio's tracking error with respect to a benchmark index. Our work differs from this line of work in that our approach involves trading derivatives and our analysis is pathwise rather than statistical. We apply the tracking methodology developed here to two important models for an equity index from the derivatives pricing literature in Sections \ref{sec:black_scholes_track} (Black-Scholes) and \ref{sec:heston} (Heston). 

Motivated still further by VIX, we explore the applications and implications of our methodology under  some  models for VIX. In particular, our methodology and examples   shed light on the connection between the mean-reverting behavior of VIX and  the implications to the pricing of VIX derivatives and tracking this index and its associated factors. We consider the tracking and risk exposure control problems under the  Cox-Ingersoll-Ross (CIR) model (see \cite{CIR85,Grunbichler1996985}), as well as the Concatenated Square Root (CSQR) model (see \cite{futuresVIX}). Among our new findings, we derive the trading strategies using options or futures to track the VIX  or achieve any exposure to the index and/or its factor (see Sections \ref{sec:one_fact_CIR} and \ref{sec:two_fact_CIR}). In addition, we illustrate the superior tracking performance of our proposed portfolio as compared to  the volatility ETN, VXX, in Section \ref{sec:VXXcomp}. This positive result can be contributed to the salient feature that our portfolio utilizes a pathwise-adaptive strategy, as opposed to a time-deterministic one used by VXX and other ETNs.  Furthermore, our strategy  is  completely explicit and can be implemented instantly. Although we have chosen VIX as our main example, our analysis applies to other mean-reverting price processes or market factors.

\newpage
\section{Continuous-Time Tracking Problem}\label{sec:cts_time_framework}
We present our tracking methodology   in  a continuous-time multi-factor diffusion framework, which  encapsulates a number of different models for financial indices and market factors. Derivatives portfolios are constructed to achieve  a pre-specified  exposure, and  their price dynamics are examined.  In particular, we compare the strategies using only futures to those using other derivatives, such as options,  for tracking and exposure control.

\subsection{Price Dynamics}\label{sec:price_dynamics}
In the background,  we fix a probability space $(\Omega, \F, \Q)$, where  $\Q$ is the risk-neutral probability  measure inferred from market derivatives prices. Throughout, we assume a   constant risk-free interest rate $r\ge 0$. Consider an index $S$, along with  $d\ge0$ exogenous observable stochastic factors $\{Y^{(i)}\,;\, i = 1, \ldots, d\}$, which satisfy the following system of stochastic differential equations (SDEs):
\begin{align}
dS_t&=\widetilde{\gamma}_t^{(0)}dt+\sigma_t^{(0,0)}dB_{t}^{\Q,0}+...+\sigma_t^{(0,d)}dB_{t}^{\Q,d}, \label{eq:index_sde} \\
dY_t^{(i)}&=\widetilde{\gamma}_t^{(i)}dt+\sigma_t^{(i,0)}dB_{t}^{\Q,0}+...+\sigma_t^{(i,d)}dB_{t}^{\Q,d}, \qquad i=1,...,d. \label{eq:factor_sde} 
\end{align} 
Here, $B_t^\Q:=(B_{t}^{\Q,0},...,B_{t}^{\Q,d})$, $t\ge 0$,  is  a $(d+1)-$dimensional standard Brownian Motion (SBM) under $\Q$.  We assume that $S_t$ and $Y_t^{(i)}$, $i=1,...,d$, are all strictly positive processes, and consider a Markovian framework whereby the coefficients are   functions of $t$, $S_t$, and $Y_t^{(i)}$, $i = 1, \ldots, d$, defined by 
\begin{equation*}
\begin{aligned}
\widetilde{\gamma}_t^{(i)}&:=\widetilde{\gamma}^{(i)}(t,S_t,Y_t^{(1)},...,Y_t^{(d)}),\quad \, i=0,...,d,\\
\sigma_t^{(i,j)}&:={\sigma}^{(i,j)}(t,S_t,Y_t^{(1)},...,Y_t^{(d)}),\quad \, i=0,...,d,\,j=0,...,d.\\
\end{aligned}
\end{equation*}
Equivalently, we can write the SDEs  more compactly in matrix form as
\begin{align*}
dM_t=\widetilde{\gamma}_t dt + \Sigma_t dB_t^\Q,
\end{align*}
where   $M_t:=(S_t,Y_t^{(1)},...,Y_t^{(d)})$ is a vector. The   drift vector is $\widetilde{\gamma}_t:=(\widetilde{\gamma}_t^{(0)},\widetilde{\gamma}_t^{(1)},...,\widetilde{\gamma}_t^{(d)})$, and the $(d+1)\times(d+1)$  volatility matrix $\Sigma_t$, has the entry $(\Sigma_t)_{i,j}=\sigma_t^{(i,j)}$ for each $i,j$. 
\begin{remark}
	The risk-neutral pricing measure, $\Q$ is an  equivalent martingale measure with respect to  the historical probability measure $\P$. The associated numeraire is  the cash account,  so all traded security prices  are discounted martingales. Under the original measure, $\P$, the market evolves according to 
	\begin{equation*}
	\begin{aligned}
	dM_t={\gamma}_t dt + \Sigma_t dB_t^\P,
	\end{aligned}
	\end{equation*}
	where $B_t^\P$ is a $(d+1)-$dimensional SBM under $\P.$ The measures $\P$ and $\Q$ are connected by the market price of risk vector $\lambda_t:=\Sigma_t^{-1}(\gamma_t-\widetilde{\gamma}_t)$. That is,
 	\begin{align}
	dB_t^\P=dB_t^\Q-\lambda_t dt. \label{dbp}
	\end{align}
 While our framework includes both complete and incomplete market models, we always assume that a risk-neutral measure has been chosen a priori and satisfies \eqref{dbp}. Since all our results and strategies are derived  and stated   \emph{pathwise} (see Propositions \ref{prop:rep_cond} and \ref{prop:pf_dynamics}), there is no need to revert from measure $\Q$ back to $\P$.
\end{remark}

Let us denote by  $c_t^{(k)} :=c^{(k)}(t,S_t,Y_t^{(1)},...,Y_t^{(d)})$, for  $k=1,...,N$, $t\in [0,T_k]$, the price processes  of $N$ European-style derivatives written on $S$, with respective terminal payoff functions $h^{(k)}(s,y_1,...,y_d)$ to be realized at time $T_k$.\footnote{We allow the payoff to depend on the factors themselves. For example, consider a spread option on an equity index $S$ and another correlated index $Y$. This amounts to setting $d=1$, and  viewing $Y$ as an index, and specifying  the option payoff     $h(s,y)=(s-y)^+$.}   At time $t\le T_k$, the no-arbitrage price of the $k$th derivative is given by 
\begin{equation}\label{eq:pricing_equation}
\begin{aligned} c^{(k)}(t,S_t,Y_t^{(1)},...,Y_t^{(d)})=\E^\Q\left[e^{-r(T_k-t)}h^{(k)}(S_{T_k},Y_{T_k}^{(1)},...,Y_{T_k}^{(d)})\big|S_t,Y_t^{(1)},...,Y_t^{(d)}\right].
\end{aligned}
\end{equation}The infinitesimal relative return of the $k$th derivative can be expressed as \begin{equation}\label{ckreturnSDE}
\frac{dc_t^{(k)}}{c_t^{(k)}}=C_t^{(k)}dt+D_t^{(k)}\frac{dS_t}{S_t}+E_{t}^{(k,1)}\frac{dY_t^{(1)}}{Y_t^{(1)}}+...+E_{t}^{(k,d)}\frac{dY_t^{(d)}}{Y_t^{(d)}},
\end{equation}
where we have defined
\begin{align}
C_t^{(k)}&:=r-\frac{\widetilde{\gamma}_t^{(0)}}{c_t^{(k)}}\frac{\partial c^{(k)}}{\partial S}-\frac{\widetilde{\gamma}_t^{(1)}}{c_t^{(k)}}\frac{\partial c^{(k)}}{\partial Y^{(1)}}-...-\frac{\widetilde{\gamma}_t^{(d)}}{c_t^{(k)}}\frac{\partial c^{(k)}}{\partial Y^{(d)}},\label{eq:nonlin_elastDef}\\
D_t^{(k)}&:=\frac{S_t}{c_t^{(k)}}\frac{\partial c^{(k)}}{\partial S},\label{eq:nonlin_elastDef2}\\
E_{t}^{(k,i)}&:=\frac{Y_t^{(i)}}{c_t^{(k)}}\frac{\partial c^{(k)}}{\partial Y^{(i)}},\quad  \, i=1,...,d. \label{eq:nonlin_elastDef3}
\end{align}
The first coefficient, $C_t^{(k)}$, is the drift of  the $k$th derivative, $D_t^{(k)}$ is the price elasticity of the $k$th derivative   with respect to the underlying index, and   $E_{t}^{(k,i)}$ is the price elasticity of the $k$th derivative   with respect to the $i$th factor. The full derivation of \eqref{ckreturnSDE} can be found in Appendix \ref{App1}.

To track an index, the trader  seeks to construct a portfolio and  precisely set  the portfolio's drift and exposure coefficients with respect to the index and its driving factors. As we will discuss next, these portfolio features  will be expressed as a linear combination of the above price elasticities. Therefore, to attain the desired exposures, the strategy is derived by solving a linear system over time. 

\subsection{Tracking Portfolio Dynamics}\label{sec:track_pf_dyn}
Fix a trading horizon $[0,T]$, with  $T\le T_k$, for all $k=1,...,N$. We  construct a self-financing portfolio,  $(X_t)_{0\le t \le T}$, utilizing $N$   derivatives with prices given by   \eqref{eq:pricing_equation}. The portfolio strategy is denoted by  the vector $w_t := (w_t^{(1)}, \ldots, w_t^{(N)})$, for $0 \le t \le T$  so that  $X_tw_t^{(k)}$ is the cash amount   invested in the $k$th derivative at time $t$ for $k = 1, \ldots, N$. Therefore, the amount  $X_t-\sum_{k=1}^{N}X_tw_t^{(k)}$ is invested at the risk-free rate $r$ at time $t$. Given such strategies, the dynamics of $X$ are 
\begin{equation}\label{eq:portfolio_val}
\begin{aligned}
\frac{dX_t}{X_t}&=\sum_{k=1}^{N}w_t^{(k)}\frac{dc_t^{(k)}}{c_t^{(k)}}+r\left(1-\sum_{k=1}^{N}w_t^{(k)}\right)dt\\
&=\left(r+\sum_{k=1}^{N}w_t^{(k)}\left(C_t^{(k)}-r\right)\right)dt+\left(\sum_{k=1}^{N}w_t^{(k)}D_t^{(k)}\right)\frac{dS_t}{S_t}+\sum_{i=1}^{d}\left(\sum_{k=1}^{N}w_t^{(k)}E_{t}^{(k,i)}\right)\frac{dY_t^{(i)}}{Y_t^{(i)}}.
\end{aligned}
\end{equation}
The three terms in \eqref{eq:portfolio_val} represent respectively the portfolio drift, exposure to the index, and exposures to the $d$  factors. 
Suppose that the investor  has chosen (i) a drift  process $(\alpha_t)_{0\le t \le T}$, (ii)  dynamic exposure coefficient $(\beta_t)_{0\le t \le T}$ with respect to the returns of $S$, and (iii)  dynamic exposure coefficients $(\eta^{(1)}_t)_{0\le t \le T},...,(\eta^{(d)}_t)_{0\le t \le T}$ with respect to the $d$ factors. Such coefficients must be adapted to the filtration in order to be attainable, but of course they may simply be constant. Then, in order to match the coefficients as desired, we must choose the strategies so as to solve the following linear system:
\begin{equation}\label{eq:linear_sys}
\left( \begin{array}{c}
\alpha_t-r \\
\beta_t\\
\eta_t^{(1)}\\
...\\
\eta_t^{(d)} \end{array} \right)= \left( \begin{array}{ccc}
C_t^{(1)}-r & ... & C_t^{(N)}-r \\
D_t^{(1)} & ... & D_t^{(N)}\\
E_{t}^{(1,1)} & ... & E_{t,1}^{(N,d)}\\
& ... &\\
E_{t}^{(1,d)} & ... & E_{t}^{(N,d)}\\ \end{array}\right)
\left( \begin{array}{c}
w_t^{(1)} \\
...\\
w_t^{(N)}\end{array} \right).
\end{equation}
However, we observe  from \eqref{eq:nonlin_elastDef}-\eqref{eq:nonlin_elastDef3}  that $C_t^{(k)}$, $D_t^{(k)}$ and $E_{t}^{(k,i)}$ satisfy 
\begin{equation}
\begin{aligned}
C_t^{(k)}-r+\frac{\widetilde{\gamma}_t^{(0)}}{S_t}D_t^{(k)}+\frac{\widetilde{\gamma}_t^{(1)}}{Y_t^{(1)}}E_{t}^{(k,1)}+...+\frac{\widetilde{\gamma}_t^{(d)}}{Y_t^{(d)}}E_{t}^{(k,d)}=0,
\end{aligned}
\end{equation}
for each $k$. Therefore, the rows of the coefficient matrix on the right hand side of   \eqref{eq:linear_sys} are  \emph{linearly dependent}. We arrive at the following proposition:
\begin{proposition}\label{prop:rep_cond}A necessary condition for the derivatives portfolio in \eqref{eq:portfolio_val} to have drift $\alpha_t$, and exposure coefficients $(\beta_t$, $\eta^{(1)}_t,$ $...,\eta^{(d)}_t)$ with respect   to $(S_t$, $Y^{(1)}_t,...,Y^{(d)}_t)$ defined in \eqref{eq:index_sde} and \eqref{eq:factor_sde} for $0\le t\le T$ is
\begin{equation}\label{eq:track_cond}
\begin{aligned}
\alpha_t-r+\frac{\widetilde{\gamma}_t^{(0)}}{S_t}\beta_t+\frac{\widetilde{\gamma}_t^{(1)}}{Y_t^{(1)}}\eta_t^{(1)}+...+\frac{\widetilde{\gamma}_t^{(d)}}{Y_t^{(d)}}\eta_t^{(d)}=0,
\end{aligned}
\end{equation}
for all $t\in[0,T]$.
\end{proposition}

We call condition \eqref{eq:track_cond} the \textbf{tracking condition}. For the general diffusion framework above,  the  left hand side of   \eqref{eq:track_cond} is    stochastic over time. It is possible to exactly control the exposures as desired almost surely if \eqref{eq:track_cond} holds pathwise.  In some  special cases,    the tracking condition \eqref{eq:track_cond} becomes a deterministic equation relating the (achievable) exposure coefficients of the factors and the index. However, one cannot expect this in general.

Instead, suppose  the dynamic exposure coefficients $(\beta_t, \eta^{(1)}_t,...,\eta^{(d)}_t)$ are pre-specified  by the trader a priori, then the tracking condition \eqref{eq:track_cond} indicates that the associated portfolio must be subject to a stochastic portfolio drift
\begin{equation}
\begin{aligned}
\alpha_t=r-\frac{\widetilde{\gamma}_t^{(0)}}{S_t}\beta_t-\frac{\widetilde{\gamma}_t^{(1)}}{Y_t^{(1)}}\eta_t^{(1)}-...-\frac{\widetilde{\gamma}_t^{(d)}}{Y_t^{(d)}}\eta_t^{(d)}.
\end{aligned}
\end{equation}
In other words,  the investor  cannot freely  control the target drift and all market exposures    simultaneously  in this general diffusion framework.  Indeed, we take this point of view in our examples to be discussed  in the following sections, and investigate  the impact of controlled exposures on the  portfolio dynamics. More generally, the tracking condition tells us that amongst the $d+2$ sources of evolution for the portfolio (drift, and $d+1$ market variables), we can only select coefficients for $d+1$ of them (unless the above condition happens to hold for that particular model.)

The tracking condition \eqref{eq:track_cond} implies that the linear system for the tracking strategies has (at least) one redundant equation. We effectively have a $(d\!+\!1)$-by-$N$ system, so  using $N>d+1$ derivatives is unnecessary and yields infinitely many portfolios having the same desired path properties. However, using exactly $N=d+1$ derivatives leads to a unique  strategy   and gives  the desired path properties. 
\begin{remark}
We have  explained one source of redundancy that exists in \emph{any} diffusion model. However,   other potential redundancies can arise depending on the derivative types and their dependencies on the index and   factors. In fact, 	it is possible that the chosen set of derivatives does not allow the resulting system  to have a unique solution. To see this, we   provide an example   in  the Heston Model in Section \ref{sec:heston} (see Example \ref{ex:heston_futures}). This is remedied by including   a  new derivative type in the portfolio (see  Example \ref{ex:heston_futures_plusVol}). 
\end{remark}

Given that the index and factor exposure coefficients are constant and that the tracking condition holds, we can derive the portfolio dynamics explicitly and illustrate  a stochastic  divergence between the portfolio   return and   targeted return. 

\begin{proposition}\label{prop:pf_dynamics}
Given constant exposure coefficients, i.e. $\beta_t=\beta$, $\eta_t^{(1)}=\eta_1$, $\ldots, \eta_t^{(d)}=\eta_d$, $\forall t\in[0,T]$, and  a strategy  $(w_t^{(1)},\ldots, w_t^{(N)})_{0 \le t \le T}$ that solves system \eqref{eq:linear_sys} , then the portfolio value defined in  \eqref{eq:portfolio_val}  admits the expression
\begin{equation}\label{eq:prop_pf_val|}
\begin{aligned}
\frac{X_u}{X_t}=\left(\frac{S_u}{S_t}\right)^\beta\prod_{i=1}^{d}\left(\frac{Y_u^{(i)}}{Y_t^{(i)}}\right)^{\eta_i}e^{\int_t^u Z_v dv},
\end{aligned}
\end{equation}
for all $0\le t \le u \le T$, where $S$, and $(Y^{(1)},...,Y^{(d)})$ satisfy    \eqref{eq:index_sde} and \eqref{eq:factor_sde}, and 
 \begin{align}
Z_t&:= \alpha_t +\frac{1}{2}\beta(1-\beta)\sum_{j=0}^d\left(\frac{\sigma_t^{(0,j)}}{S_t}\right)^2  +\frac{1}{2}\sum_{i=1}^d\sum_{j=0}^d\eta_i(1-\eta_i)\left(\frac{\sigma_t^{(i,j)}}{Y_t^{(i)}}\right)^2  \label{zzzz}\\
&-\beta\sum_{i=1}^{d}\sum_{j=0}^{d}\eta_i\frac{\sigma_t^{(i,j)}}{Y_t^{(i)}}\frac{\sigma_t^{(0,j)}}{S_t}-\sum_{i=1}^{d}\sum_{l=i+1}^{d}\sum_{j=0}^{d}\eta_i\eta_l\frac{\sigma_t^{(i,j)}}{Y_t^{(i)}}\frac{\sigma_t^{(l,j)}}{Y_t^{(l)}}.\,\notag
\end{align}
 \end{proposition}
\begin{proof}
Applying  Ito's formula, we write down the SDEs for  $d\log(S_t)$ and  $d\log(Y_t^{(i)})$, for  each $i$:
 \begin{align}
 d\log(S_t)&=\frac{dS_t}{S_t}-\frac{1}{2}\sum_{j=0}^{d}\left(\frac{\sigma_t^{(0,j)}}{S_t}\right)^2dt,\\
 d\log(Y_t^{(i)})&=\frac{dY_t^{(i)}}{Y_t^{(i)}}-\frac{1}{2}\sum_{j=0}^{d}\left(\frac{\sigma_t^{(i,j)}}{Y_t^{(i)}}\right)^2dt.
\end{align}
Next we multiply   $d\log(S_t)$ by $\beta$ and each  $d\log(Y_t^{(i)})$ by $\eta_i$ respectively, and add these all together with the term $\alpha_t dt$ to obtain
\begin{equation}\label{eq:proofEq}
\begin{aligned}
\alpha_t dt &+ \sum_{j=1}^{d}\eta_jd\log(Y_t^{(j)}) + \beta d\log(S_t)=\frac{dX_t}{X_t}-\frac{1}{2}\sum_{i=1}^{d}\sum_{j=0}^{d}\eta_i\left(\frac{\sigma_t^{(i,j)}}{Y_t^{(i)}}\right)^2\!dt-\frac{1}{2}\beta\sum_{j=0}^{d}\left(\frac{\sigma_t^{(0,j)}}{S_t}\right)^2\!dt.
\end{aligned}
\end{equation}
Now, apply  Ito's formula to $\log(X_t)$, we have
\begin{equation*}
\begin{aligned}
\frac{dX_t}{X_t}=d\log(X_t)+\frac{1}{2}\left(\frac{dX_t}{X_t}\right)^2.
\end{aligned}
\end{equation*} 
Applying this to \eqref{eq:proofEq}, we get 
\begin{align}
\alpha_t dt &+ \sum_{j=1}^{d}\eta_jd\log(Y_t^{(j)}) + \beta d\log(S_t)
\notag\\
&=d\log(X_t)+\frac{1}{2}\left(\frac{dX_t}{X_t}\right)^2-\frac{1}{2}\sum_{i=1}^{d}\sum_{j=0}^{d}\eta_i\left(\frac{\sigma_t^{(i,j)}}{Y_t^{(i)}}\right)^2dt
-\frac{1}{2}\beta\sum_{j=0}^{d}\left(\frac{\sigma_t^{(0,j)}}{S_t}\right)^2dt.\label{eq:proofEq2}
\end{align}
Since   there exists strategy $(w_t)_{0 \le t \le T}$ that solves   \eqref{eq:linear_sys}, the portfolio   in   \eqref{eq:portfolio_val} satisfies the SDE 
\begin{equation}\label{eq:proof_eq_sdeX}
	\begin{aligned}
		\frac{dX_t}{X_t}=\alpha_t dt + \beta\frac{dS_t}{S_t}+\sum_{i=1}^{d}\eta_i\frac{dY_t^{(i)}}{Y_t^{(i)}}.
	\end{aligned}
\end{equation} 
By squaring    \eqref{eq:proof_eq_sdeX}, we have
 \begin{align}
\frac{1}{2}\left(\frac{dX_t}{X_t}\right)^2&=\frac{1}{2}\beta^2\left(\frac{dS_t}{S_t}\right)^2+\frac{1}{2}\sum_{i=1}^{d}\eta_i^2\left(\frac{dY_t^{(i)}}{Y_t^{(i)}}\right)^2+\sum_{i=1}^{d}\beta\eta_i\frac{dY_t^{(i)}}{Y_t^{(i)}}\frac{dS_t}{S_t}+\sum_{i=1}^{d}\sum_{l=i+1}^{d}\eta_i\eta_l\frac{dY_t^{(i)}}{Y_t^{(i)}}\frac{dY_t^{(l)}}{Y_t^{(l)}}\notag\\
&=\frac{1}{2}\beta^2\sum_{j=0}^{d}\left(\frac{\sigma_t^{(0,j)}}{S_t}\right)^2dt+\frac{1}{2}\sum_{i=1}^{d}\sum_{j=0}^{d}\eta_i^2\left(\frac{\sigma_t^{(i,j)}}{Y_t^{(i)}}\right)^2dt+\beta\sum_{i=1}^{d}\sum_{j=0}^{d}\eta_i\frac{\sigma_t^{(i,j)}}{Y_t^{(i)}}\frac{\sigma_t^{(0,j)}}{S_t}dt\notag\\
&\quad +\sum_{i=1}^{d}\sum_{l=i+1}^{d}\sum_{j=0}^{d}\eta_i\eta_l\frac{\sigma_t^{(i,j)}}{Y_t^{(i)}}\frac{\sigma_t^{(l,j)}}{Y_t^{(l)}}dt\notag\\
&=\alpha_t dt+\frac{1}{2}\beta\sum_{j=0}^{d}\left(\frac{\sigma_t^{(0,j)}}{S_t}\right)^2dt+
\frac{1}{2}\sum_{i=1}^{d}\sum_{j=0}^{d}\eta_i\left(\frac{\sigma_t^{(i,j)}}{Y_t^{(i)}}\right)^2dt-Z_tdt,\label{halfdx}
\end{align}
where $Z_t$ is defined in \eqref{zzzz}. Next, substituting  \eqref{halfdx} into   \eqref{eq:proofEq2} and rearranging, we obtain  the  link  between the   log returns of the portfolio and those of the index and factors:
\begin{equation}
 d\log(X_t)=\sum_{j=1}^{d}\eta_jd\log(Y_t^{(j)}) + \beta d\log(S_t)+Z_tdt.\label{dlogXsde}
 \end{equation}
Upon integrating \eqref{dlogXsde} and exponentiating, we obtain the desired result.
\end{proof}

We call the stochastic process $Z_t$ in \eqref{zzzz} the \textbf{slippage process} as it is typically negative and describes the deviations of the portfolio returns from the targeted returns. In particular, taking the logarithm  in \eqref{eq:prop_pf_val|} gives us the relationship between the portfolio's log return and the log returns of the index and each of the factors over any given period $[t,u]$, namely, 
\begin{equation*}
\begin{aligned}
\log\left(\frac{X_u}{X_t}\right)=\beta\log\left(\frac{S_u}{S_t}\right)+\sum_{j=1}^{d}\eta_i\log\left(\frac{Y_u^{(i)}}{Y_t^{(i)}}\right)+\int_t^u Z_v dv.
\end{aligned}
\end{equation*}
The first two terms  indicate that the portfolio's log return is proportional to the log returns of the index and its driving factors, with the proportionality coefficients being equal to the desired exposure. However, the portfolio's  log return   is subject to  the integrated slippage process. 

Integrating the square of the volatility of the index yields   the \emph{realized variance}. As such, the  portfolio value, $X$, is akin to the price  process of an LETF (see e.g. \cite{AZ})  in continuous time, except that $X$ also controls the exposure to various factors in addition to maintaining a fixed leverage ratio $\beta$ w.r.t. $S$.  Our framework allows for a multidimensional model, so  the portfolio value and   the realized slippage depend not only on the realized variances of the underlying factors, but also the  \emph{realized covariances} between the index and   factors. 

\begin{remark}
If we apply the notations,  $Y_t^{(0)} \equiv S_t$  and $\eta_0\equiv \beta$, then the portfolio value   simplifies to 
\begin{equation*}
\begin{aligned}
\frac{X_T}{X_t}=\prod_{i=0}^{d}\left(\frac{Y_T^{(i)}}{Y_t^{(i)}}\right)^{\eta_j}e^{\int_t^T Z_u du},
\end{aligned}
\end{equation*}
and the slippage process  admits a more compact expression
\begin{equation*}
\begin{aligned}
Z_t=\alpha_t+\frac{1}{2}\sum_{i=0}^d\sum_{j=0}^d\eta_i(1-\eta_i)\left(\frac{\sigma_t^{(i,j)}}{Y_t^{(i)}}\right)^2-\sum_{i=0}^d\sum_{l=i+1}^d\sum_{j=0}^d\eta_i\eta_l\frac{\sigma_t^{(i,j)}}{Y_t^{(i)}}\frac{\sigma_t^{(l,j)}}{Y_t^{(l)}}.
\end{aligned}
\end{equation*}
\end{remark}

\begin{example}
If there is no additional stochastic factor ($d=0$), then we have
\begin{equation*}
\begin{aligned}
\frac{X_T}{X_t}=\left(\frac{S_T}{S_t}\right)^\beta e^{\int_t^T Z_u du},
\end{aligned}
\end{equation*}
and
\begin{equation}\label{eq:zzzz_0}
Z_t:=\alpha_t+\frac{1}{2}\beta(1-\beta)\left(\frac{\sigma_t^{(0,0)}}{S_t}\right)^2,
\end{equation}
where $\sigma_t^{(0,0)}$ can depend on  $t$ and   $S_t$ as in the   general local volatility framework.  As such, the   slippage process does not involve a   covariance term, but it  reflects the  volatility decay, which is well documented for leveraged ETFs with an integer $\beta$  (see e.g. \cite{AZ} and \cite{LeungSantoliBook}).  As seen in \eqref{eq:zzzz_0}, there is an  erosion of   the portfolio (log-)return  that is  proportional to the realized variance of the index whenever $\beta\notin[0,1]$. 
\end{example}

\begin{example}
	If $d=1$, i.e. one exogenous market factor $Y$ along with the index $S$, we have
	\begin{equation*}
	\begin{aligned}
	\frac{X_T}{X_t}=\left(\frac{S_T}{S_t}\right)^\beta \left(\frac{Y_T}{Y_t}\right)^\eta e^{\int_t^T Z_u du},
	\end{aligned}
	\end{equation*}
	and
	\begin{equation}\label{eq:zt_1hiddenfactor}
	\begin{aligned}
	Z_t:=\alpha&+\frac{1}{2}\beta(1-\beta)\left(\frac{\sigma_t^{(S)}}{S_t}\right)^2+\frac{1}{2}\eta(1-\eta)\left(\frac{\sigma_t^{(Y)}}{Y_t}\right)^2-\beta\eta\frac{\sigma_t^{(S,Y)}}{S_tY_t},\\
	\end{aligned}
	\end{equation}
	where
	\begin{equation*}
	\begin{aligned}
	\sigma_t^{(S)}&:=\sqrt{\left(\sigma_t^{(0,0)}\right)^2+\left(\sigma_t^{(0,1)}\right)^2},\\
	\sigma_t^{(Y)}&:=\sqrt{\left(\sigma_t^{(1,0)}\right)^2+\left(\sigma_t^{(1,1)}\right)^2},\\
	\sigma_t^{(S,Y)}&:=\sigma_t^{(1,0)}\sigma_t^{(0,0)}+\sigma_t^{(1,1)}\sigma_t^{(0,1)}.
	\end{aligned}
	\end{equation*}
Here, $\sigma_t^{(S)}$, $\sigma_t^{(Y)}$, and $\sigma_t^{(S,Y)}$, are    functions of $t$, $S_t$, and $Y_t$, as in the   Local Stochastic Volatility framework. In addition to the value erosion proportional to the realized variance of the index (second term in \eqref{eq:zt_1hiddenfactor}),  there is another  realized variance decay term for the exogenous factor (third term in \eqref{eq:zt_1hiddenfactor}). Indeed it will be negative whenever the $Y$-exposure coefficient $\eta\notin[0,1]$. 	Beyond the realized variance of the index and the factor, there is also a term in the slippage process which accounts for the realized covariance between the index and factor (final term in   \eqref{eq:zt_1hiddenfactor}). It is negative if   $\sigma_t^{(S,Y)}$, $\beta$, and $\eta$ are all positive, reflecting  another source of  value erosion relative to the desired log return. 
\end{example}

\subsection{Portfolios with Futures}\label{sec:futures_pf_track}
Futures contracts are also useful  instruments for index tracking.\footnote{See e.g. \cite{AB} for a discussion on the empirical  performance of minimum variance hedging strategies using  futures contracts against an index ETF. } The price of a futures contract of maturity $T_k$\footnote{Despite the identical notation,  the maturities of the futures can be different from those in Section \ref{sec:track_pf_dyn}.} at time $t\le T_k$ is given by
\begin{equation}\label{eq:pricing_equation_f}
\begin{aligned}
f_t^{(k)}:= f^{(k)}(t,S_t,Y_t^{(1)},...,Y_t^{(d)})=\E^\Q\left[S_{T_k}\big|S_t,Y_t^{(1)},...,Y_t^{(d)}\right].
\end{aligned}
\end{equation}
The Feynman-Kac formula implies that the price must satisfy
\begin{equation}
\label{FKfutures}
\frac{\partial f^{(k)}}{\partial t}+\widetilde{\gamma}_t^{(0)}\frac{\partial f^{(k)}}{\partial S}+\widetilde{\gamma}_t^{(1)}\frac{\partial f^{(k)}}{\partial Y^{(1)}}+...+\widetilde{\gamma}_t^{(d)}\frac{\partial f^{(k)}}{\partial Y^{(d)}}+\frac{1}{2}\Tr\left[\Sigma^\intercal \nabla^2 f^{(k)}\Sigma\right]=0,
\end{equation}
with the terminal condition$f(T_k,s,y_{1},...,y_{d})=s$ for all vectors $(s,y_1,...,y_d)$ with strictly positive components. By \eqref{FKfutures} and Ito's formula, we get 
\begin{equation*}
\begin{aligned}
df_t^{(k)}
&=\left(-\widetilde{\gamma}_t^{(0)}\frac{\partial f^{(k)}}{\partial S}-\widetilde{\gamma}_t^{(1)}\frac{\partial f^{(k)}}{\partial Y^{(1)}}-...-\widetilde{\gamma}_t^{(d)}\frac{\partial f^{(k)}}{\partial Y^{(d)}}\right)dt+\frac{\partial f^{(k)}}{\partial S}dS_t\\
&~~~+\frac{\partial f^{(k)}}{\partial Y^{(1)}}dY_t^{(1)}+...+\frac{\partial f^{(k)}}{\partial Y^{(d)}}dY_t^{(d)}.
\end{aligned}
\end{equation*}

Now consider a self-financing portfolio $(X_t)_{t\ge 0}$ utilizing $N$ futures of maturities $T_1,T_2,...,T_N$ over a trading horizon $T\le T_k$  for all $k$. Denote by $(u^{(k)}_t)_{0\le t \le T}$, $k=1,...,N$  a generic adapted strategy   such that, at time $t$, the cash amount $u_t^{(k)}X_t$ is  invested in the $k$th futures contract.\footnote{It is costless to establish a  futures position, so no borrowing is involved.} Assuming that the futures contracts are continuously marked to market, the portfolio value evolves according to 
\begin{equation}\label{eq:portfolio_val_f}
\begin{aligned}
\frac{dX_t}{X_t} &= \sum_{k=1}^N u^{(k)}_t \frac{df^{T_k}_t}{f^{(k)}_t}+rdt.\\
 &=\left(r+\sum_{k=1}^{N}u_t^{(k)}F_t^{(k)}\right)dt+\sum_{k=1}^{N}u_t^{(k)}G_t^{(k)}\frac{dS_t}{S_t}+\sum_{i=1}^{d}\sum_{k=1}^{N}u_t^{(k)}H_{t}^{(k,i)}\frac{dY_t^{(i)}}{Y_t^{(i)}},
\end{aligned}
\end{equation}
where   coefficients are defined by
 \begin{align}
F_t^{(k)}&:=-\frac{\widetilde{\gamma}_t^{(0)}}{f_t^{(k)}}\frac{\partial f^{(k)}}{\partial S}-\frac{\widetilde{\gamma}_t^{(1)}}{f_t^{(k)}}\frac{\partial f^{(k)}}{\partial Y^{(1)}}-...-\frac{\widetilde{\gamma}_t^{(d)}}{f_t^{(k)}}\frac{\partial f^{(k)}}{\partial Y^{(d)}}\,, \label{EQF}\\
G_t^{(k)}&:=\frac{S_t}{f_t^{(k)}}\frac{\partial f^{(k)}}{\partial S}\,,\label{EQG}\\
H_{t}^{(k,i)}&:=\frac{Y_t^{(i)}}{f_t^{(k)}}\frac{\partial f^{(k)}}{\partial Y^{(i)}}, \quad  i=1,...,d.&\label{EQH}
\end{align}
Comparing \eqref{EQF} to \eqref{eq:nonlin_elastDef}, we notice that $F_t^{(k)}$ has no $r$ term.

Again suppose the trader selects a dynamic exposure coefficient $(\beta_t)_{0\le t \le T}$ with respect to the returns of $S$, as well as dynamic exposure coefficients for each factor return, $(\eta^{(1)}_t, \ldots, \eta^{(d)}_t)_{0\le t \le T}$. Suppose further that the trader chooses a target dynamic drift $(\alpha_t)_{0\le t \le T}$. In order for the portfolio to attain these desired path properties, we must solve the following linear system
\[\left( \begin{array}{c}
\alpha_t-r \\
\beta_t\\
\eta_t^{(1)}\\
...\\
\eta_t^{(d)} \end{array} \right)= \left( \begin{array}{ccc}
F_t^{(1)} & ... & F_t^{(N)} \\
G_t^{(1)} & ... & G_t^{(N)}\\
H_{t}^{(1,1)} & ... & H_{t}^{(N,1)}\\
& ... &\\
H_{t}^{(1,d)} & ... & H_{t}^{(N,d)}\\ \end{array}\right)
\left( \begin{array}{c}
u_t^{(1)} \\
...\\
u_t^{(N)}\end{array} \right).\]
The definitions of $F_t^{(k)}$, $G_t^{(k)}$, and $H_{t}^{(k,i)}$ in \eqref{EQF}-\eqref{EQH} imply that 
\begin{equation}
\begin{aligned}
F_t^{(k)}+\frac{\widetilde{\gamma}_t^{(0)}}{S_t}G_t^{(k)}+\frac{\widetilde{\gamma}_t^{(1)}}{Y_t^{(1)}}H_{t}^{(k,1)}+...+\frac{\widetilde{\gamma}_t^{(d)}}{Y_t^{(d)}}H_{t}^{(k,d)}=0,
\end{aligned}
\end{equation}
for each $k$. It follows that the rows of the linear system are linearly dependent. Thus, for the system to be consistent we must  require that 
\begin{equation}
\begin{aligned}
\alpha_t-r+\frac{\widetilde{\gamma}_t^{(0)}}{S_t}\beta_t+\frac{\widetilde{\gamma}_t^{(1)}}{Y_t^{(1)}}\eta_t^{(1)}+...+\frac{\widetilde{\gamma}_t^{(d)}}{Y_t^{(d)}}\eta_t^{(d)}=0,
\end{aligned}
\end{equation} for all $t\in[0,T]$.  As it turns out, this  is the same tracking condition    in   Proposition \ref{prop:rep_cond}, so the ensuing   discussion applies to the current case with futures as well.  However,  the tracking strategies associated with futures can be significantly different from those with other derivatives.

\section{Equity Index Tracking}\label{sec:equity_idx}
In this section we discuss two prominent equity derivatives pricing  models captured by our  framework and  present  the tracking conditions and strategies. 

\subsection{Black-Scholes Model}\label{sec:black_scholes_track}
We consider tracking using derivatives  on an underlying  index $S$ under  the Black-Scholes model, so there is no additional exogenous factors ($d=0$).  Under the risk-neutral measure, the index follows  
\begin{equation}\label{eq:bs_dynamics}
\begin{aligned}
dS_t=rS_tdt+\sigma S_t dB_t^\Q,
\end{aligned}
\end{equation}
where $B_t^\Q$ is a SBM and  $\sigma>0$ is the volatility parameter. Then, applying Proposition \ref{prop:rep_cond}, the tracking condition under the  Black-Scholes model is simply
\begin{equation}\label{eq:bs_replicating}
\begin{aligned}
\alpha_t=r(1-\beta_t), \qquad 0 \le t \le T.
\end{aligned}
\end{equation}

A few remarks are in order. First, a zero exposure ($\beta_t = 0$) implies that  the portfolio grows at   the risk free rate ($\alpha_t= r$).   If $\beta_t=1$, then $\alpha_t=0$, which means that perfect tracking  of $S$ is possible with no excess drift. This portfolio is a full investment in the index via some derivative. According to the tracking condition \eqref{eq:bs_replicating}, if $\beta_t>1$, then $\alpha_t<0$. This indicates that borrowing is required in order to leverage the underlying returns. Moreover, we have $\alpha_t>r$ as long as  $\beta_t<0$. Hence, by shorting the index, one achieves a drift above the risk-free rate. For any value of $\beta_t$ between 0 and 1, the strategy is trading off an investment in the money market account and the underlying index (via a derivative security).

Now suppose for the rest of this subsection that the drift and exposure coefficient are constant, namely, $\alpha_t\equiv\alpha$ and $\beta_t\equiv \beta$.  More specifically, the trader specifies the exposure to $S$ by setting  the value of $\beta$ so that condition \eqref{eq:bs_replicating} implies the fixed drift $\alpha=r(1-\beta)$. By combining Proposition \ref{prop:pf_dynamics} with condition \eqref{eq:bs_replicating},  the  portfolio value can be expressed explicitly as 
\begin{equation}\label{eq:BS_equity_val}
\begin{aligned}
X_t=X_0\left(\frac{S_t}{S_0}\right)^\beta e^{(r+\frac{\beta\sigma^2}{2})(1-\beta)t},
\end{aligned}
\end{equation}
or equivalently, in terms of log-returns
\begin{equation*}
\begin{aligned}
\log\left(\frac{X_t}{X_0}\right)=\beta\log\left(\frac{S_t}{S_0}\right)+ (r+\frac{\beta\sigma^2}{2})(1-\beta)t.
\end{aligned}
\end{equation*}
In particular, the slippage process is a constant, given by 
\begin{equation*}
\begin{aligned}
Z_t=(r+\frac{\beta\sigma^2}{2})(1-\beta).
\end{aligned}
\end{equation*} It is also  quadratic concave in $\beta$.   It follows that the slippage  is non-negative for $\beta \in [-{2r}{\sigma^{-2}},  1]$, and is strictly  negative  otherwise. Therefore, for constant exposure coefficient $\beta$ outside (resp. inside)  of the interval, $[-{2r}{\sigma^{-2}},  1]$, the log-return of the tracking portfolio is lower  (resp. higher) than  the corresponding multiple $(\beta)$ of the index's log-return. 

To better understand the slippage, we take $r=0.05$ and $\sigma=0.2$. Then, for any $\beta \notin [-2.5,1]$, the tracking portfolio's log-return falls short of the respective multiple of the index's log-return. To illustrate this, we display in Figure \ref{fig:BS_tracking_error} the simulated sample paths of the portfolio values, along with their respective benchmark (whose log return is equal to the respective multiple of the index's log return),  for $\beta\in\left\{-1,2,3\right\}$. As expected, when $\beta = 2$ or $3$, the portfolio underperforms compared to the benchmark. In contrast, the portfolio outperforms the benchmark when  $\beta=-1 \in [-2.5,1]$. 

\begin{figure}
	\begin{centering}
		%		\subfigure[$\beta=-4$]{\includegraphics[trim={1.2cm 0.5cm 1.5cm 0.5cm},clip,width=3in]{BSbetaNEG4track}}
		\subfigure[$\beta=-1$]{\includegraphics[trim={1.2cm 0.5cm 1.5cm 0.3cm},clip,width=5in]{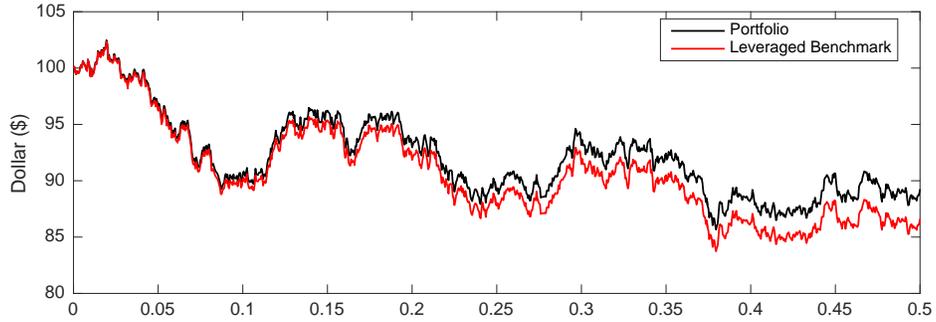}}
		\subfigure[$\beta=2$]{\includegraphics[trim={1.2cm 0.5cm 1.5cm 0.3cm},clip,width=5in]{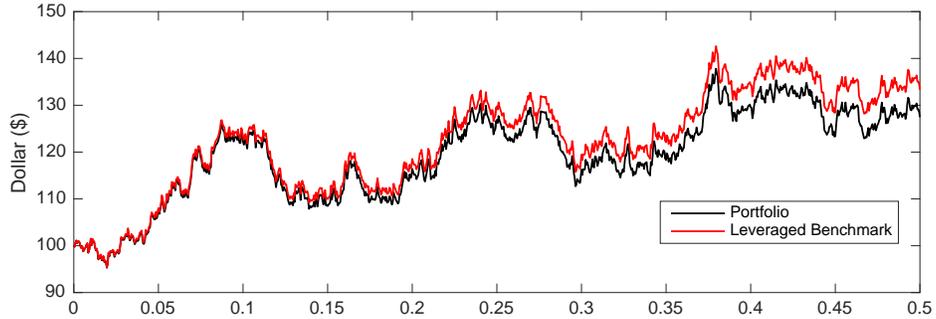}}
		\subfigure[$\beta=3$]{\includegraphics[trim={1.2cm 0.5cm 1.5cm 0.3cm},clip,width=5in]{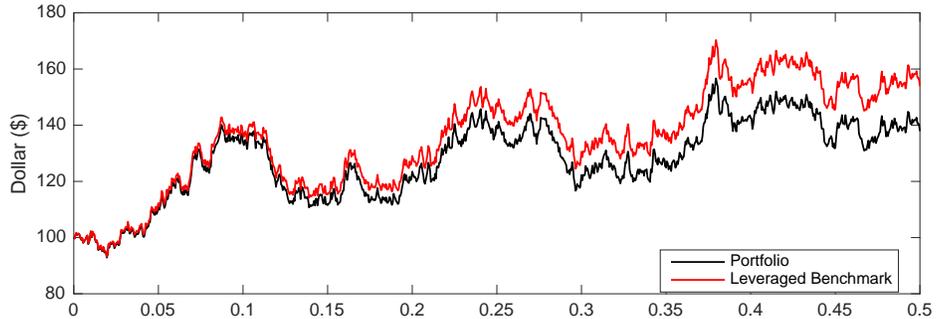}}
		\caption{\small{Sample paths  of portfolio values (in \$) compared to their benchmarks for (a) $\beta=-1$, (b) $\beta=2$, and (c) $\beta=3$ under the Black-Scholes model. Parameters: $X_0=100$, $S_0$, $r=0.05$,   $\sigma=0.2$, and $T=0.5$. The benchmark is defined such that its log-return is equal to $\beta$ times the index's log-return, with an initial value of $\$100$.}}
		\label{fig:BS_tracking_error}
	\end{centering}
\end{figure}

The associated  strategy achieving such path properties requires the use of at least $d+1=1$ derivative. Using exactly $1$ leads to a unique strategy. To find the unique strategy, we can (without loss of generality) solve the corresponding equation in \eqref{eq:linear_sys} to get $w_t =\beta/D_t$. Let us compare   the tracking strategies using  call options and futures contracts. First, consider a call on the index with expiration date  $T_c$ and  strike $K$. Its price is given by the \cite{blackscholes} formula 
\begin{equation*}
\begin{aligned}
c_t:=c(t,S_t)=S_tN\left(d_+(t,S_t)\right)-Ke^{-r(T_c-t)}N\left(d_-(t,S_t)\right),
\end{aligned}
\end{equation*}
where 
\begin{equation*}
\begin{aligned}
d_\pm(t,S_t)&=\frac{\log\left(S_t/K\right)+(r\pm\frac{\sigma^2}{2})(T_c-t)}{\sigma\sqrt{T_c-t}},\quad N(x)=\int_{-\infty}^{x}\frac{1}{\sqrt{2\pi}}e^{-\frac{u^2}{2}}du.
\end{aligned}
\end{equation*}
Therefore, to obtain an exposure coefficient of $\beta$ to the index returns, requires holding
\begin{equation}\label{eq:call_strat}
\begin{aligned}
\frac{w_tX_t}{c_t}=\frac{\beta X_0}{S_0}\left(\frac{S_t}{S_0}\right)^{\beta-1} e^{(r+\frac{\beta\sigma^2}{2})(1-\beta)t}\frac{1}{N\left(d_+(t,S_t)\right)}
\end{aligned}
\end{equation}
units of call option at time $t$.

The  price of a  futures  written  on $S$ with maturity  $T_f$ is given by $f_t:=f(t,S_t)=S_te^{r(T_f-t)}$ for $t\le T_f$. Therefore, to achieve  an exposure with coefficient $\beta$  requires holding
\begin{equation}\label{eq:futs_strat}
\begin{aligned}
\frac{w_tX_t}{f_t}= \frac{\beta X_0}{S_0}\left(\frac{S_t}{S_0}\right)^{\beta-1} e^{(r+\frac{\beta\sigma^2}{2})(1-\beta)t-r(T_f-t)}.
\end{aligned}
\end{equation}
  contracts at time $t$.

\begin{figure}
	\begin{centering}
		\subfigure[Options]{\includegraphics[trim={1.2cm 0.7cm 1.1cm 0.5cm},clip,width=3in]{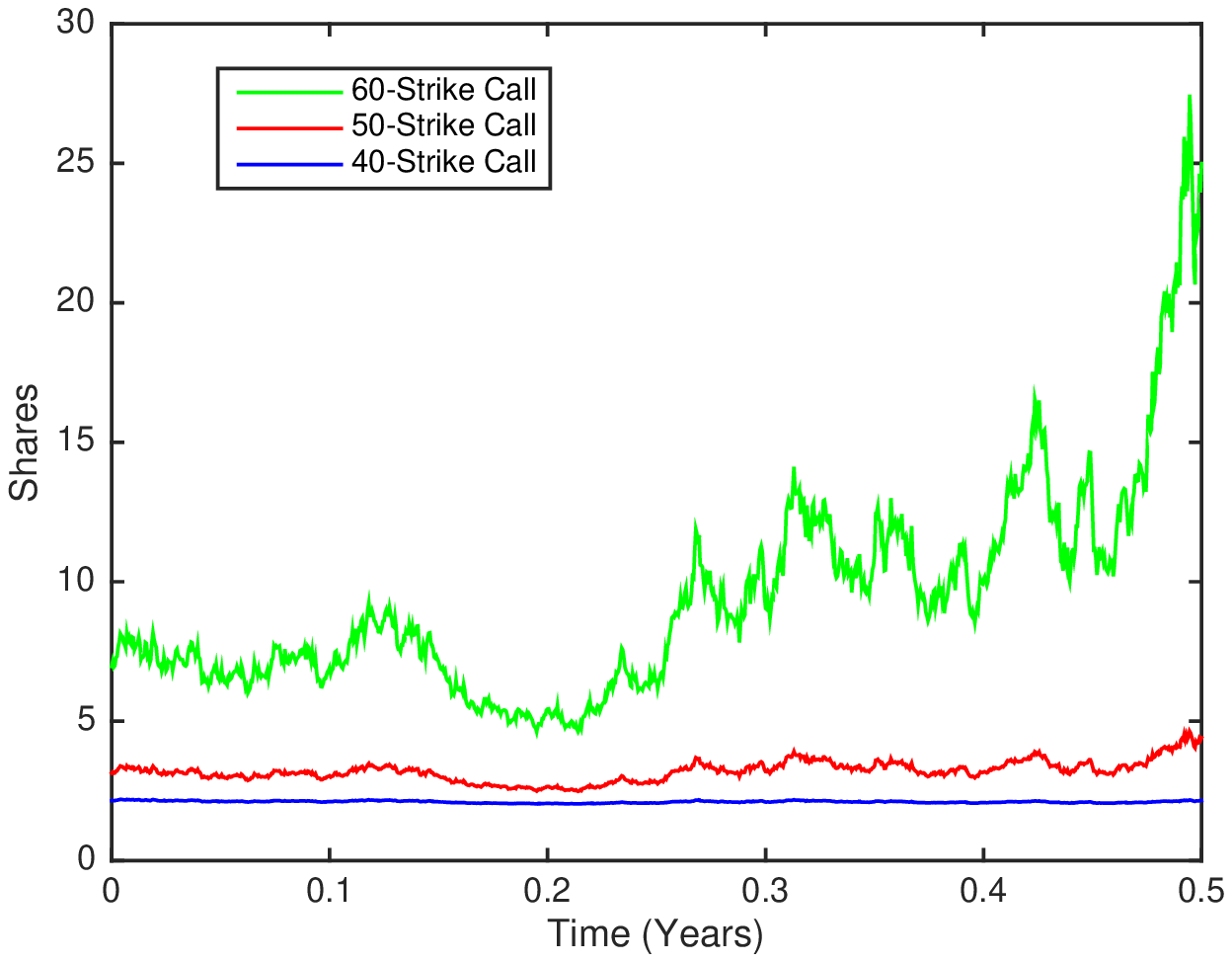}}
		\subfigure[Futures]{\includegraphics[trim={1cm 0.7cm 1.1cm 0.5cm},clip,width=3in]{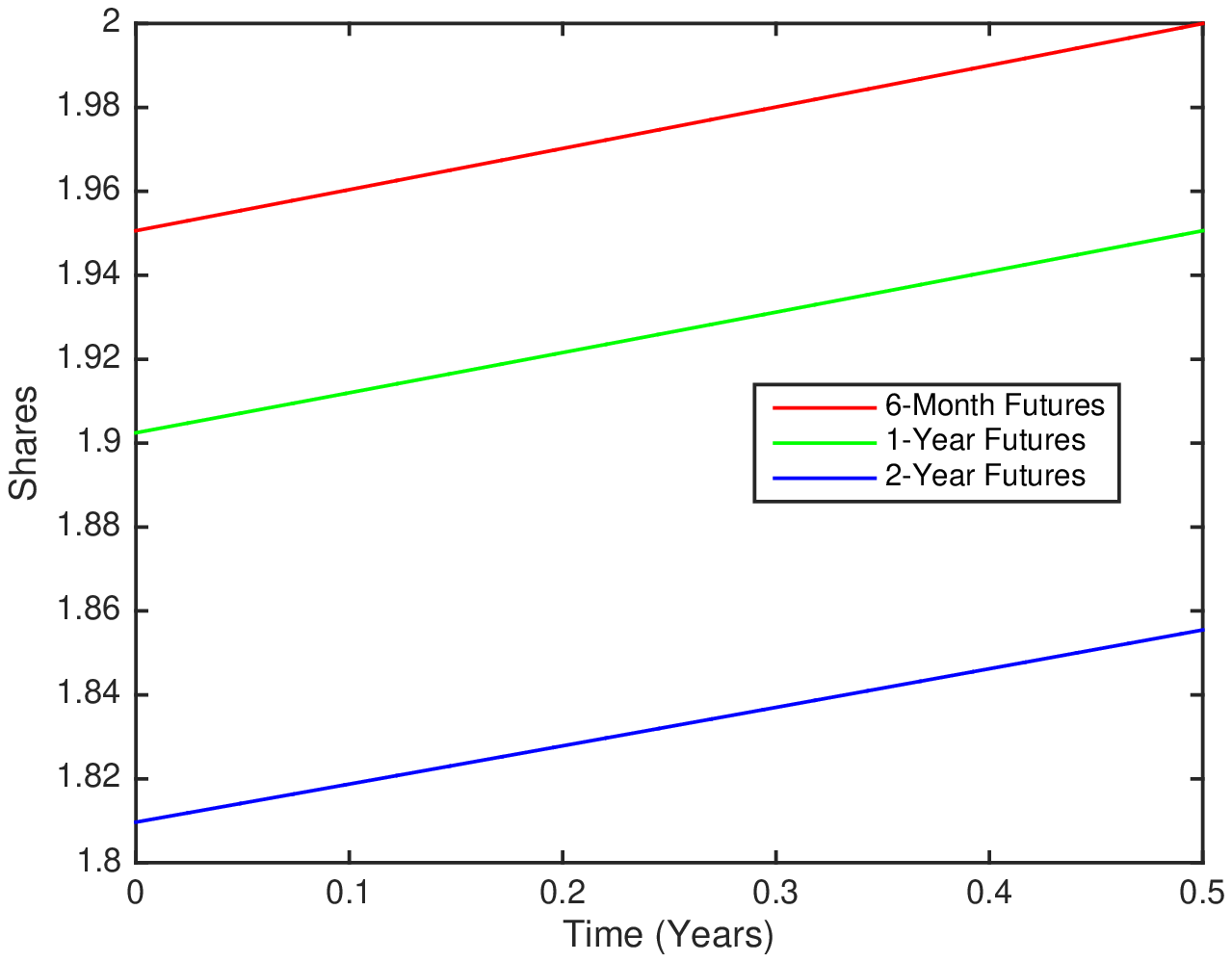}}
		\caption{\small{Simulations of  tracking strategies when using  (a)  a call option of different strike $K\in\left\{40,50,60\right\}$ and $T_c=0.5$, and (b) futures with maturity $T_f\in\left\{0.5,1,2\right\}$ under the Black-Scholes model. Parameters: $S_0=50$, $X_0=100$, $r=0.05$, $\sigma=0.2$, and $T=0.5$.}}
		\label{fig:BShedge_simulation}
	\end{centering}
\end{figure}

\begin{figure}
	\begin{centering}
		\subfigure[Options]{\includegraphics[trim={1.2cm 0.5cm 1.5cm 0.3cm},clip,width=5in]{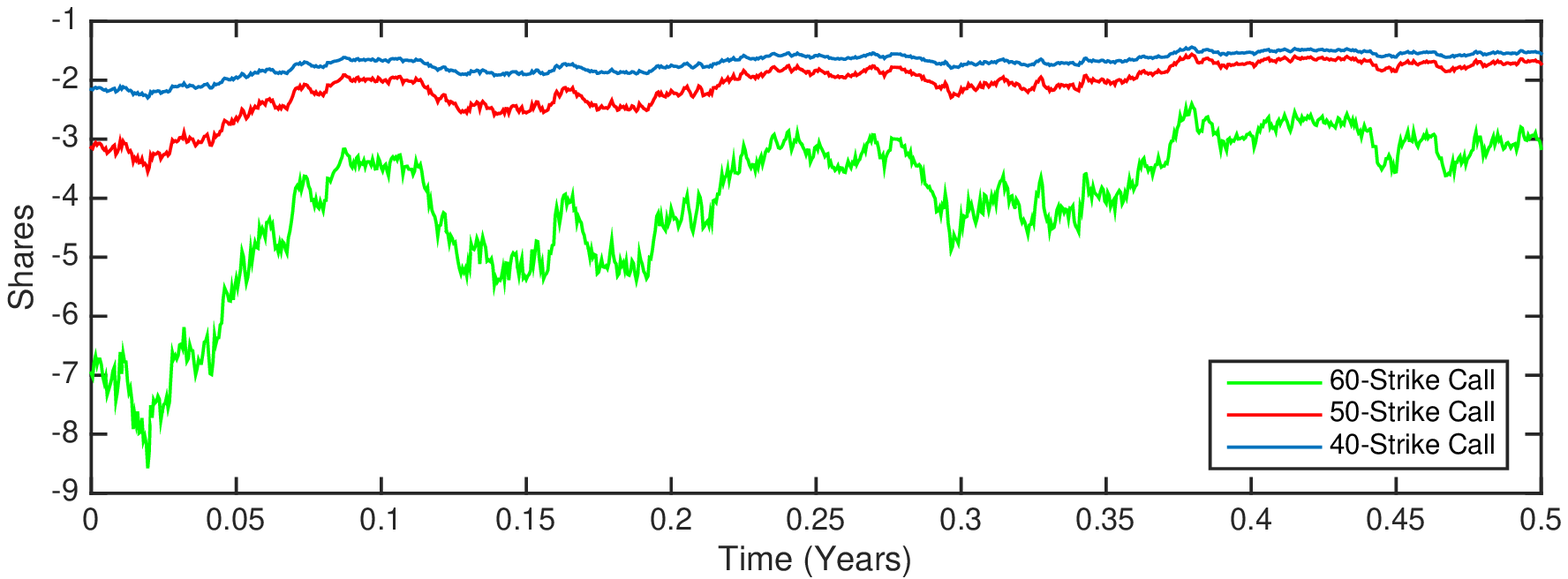}}
		\subfigure[Futures]{\includegraphics[trim={1.2cm 0.5cm 1.5cm 0.3cm},clip,width=5in]{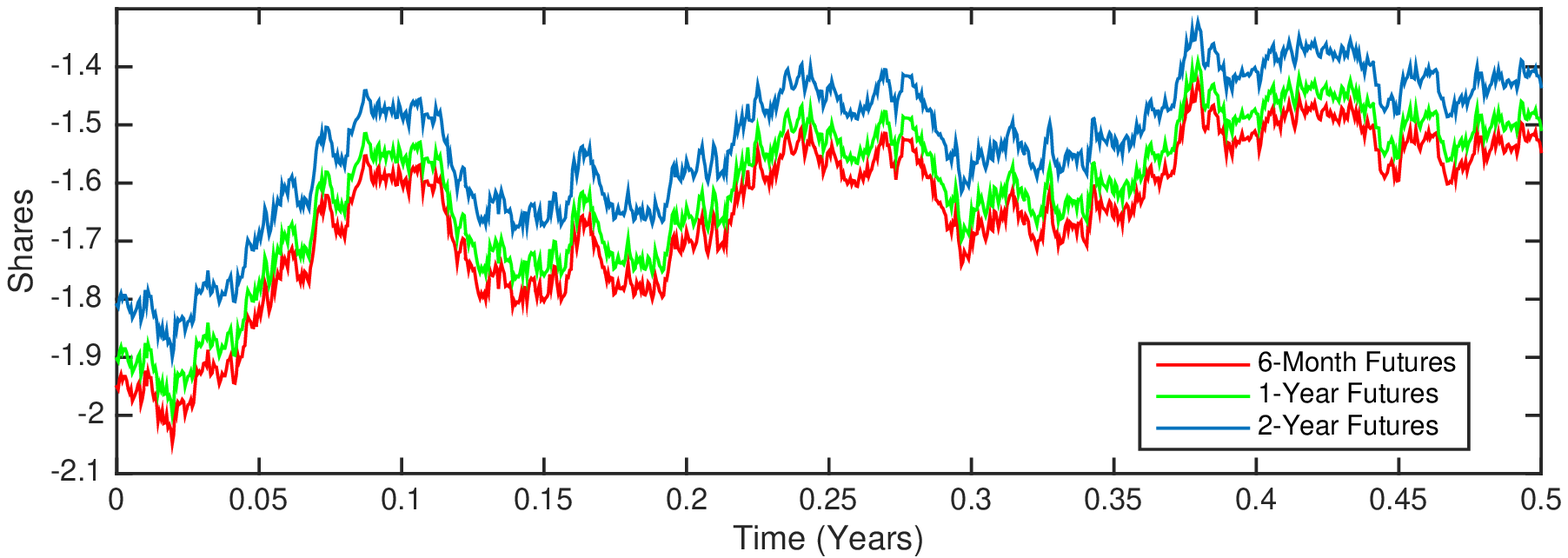}}
		\caption{\small{Simulation of tracking strategies for (a) call options of various strikes ($K\in\left\{40,50,60\right\}$ and $T_c=0.5$) and (b) futures contracts of various maturities ($T_f\in\left\{0.5,1,2\right\}$). These strategies target $\beta=-1$ over the trading horizon until $T=0.5$. Parameters are $S_0=50$, $X_0=100$, $r=0.05$ and $\sigma=0.2$. For the ease of comparison, the strategies are based on sample path of the reference index used in Figure \ref{fig:BS_tracking_error}.}}
		\label{fig:BS_tracking_strats}
	\end{centering}
\end{figure}

To gain further intuition, let us compare the above strategies when the investor seeks a  \emph{unit exposure} ($\beta=1$) with respect to $S$. The corresponding futures and options holdings are, respectively,   
\begin{equation*}
 \frac{ X_0}{S_0} e^{-r(T_f-t)}
 \qquad \text{and} \qquad \frac{ X_0}{S_0N\left(d_+(t,S_t)\right)}\,.
 \end{equation*}
First, these strategies can be viewed as the reciprocal of the associated delta hedge. Under this model, both calls and futures allow for perfect tracking of $S$, since  $\beta=1$ implies that  $\alpha=0$.  However, the two   strategies are very different. 

With a call option, the strategy is stochastic, depending crucially on the index dynamics. In Figure \ref{fig:BShedge_simulation}(a), we  display  the hedging strategies for call options of strikes $K\in\left\{40,50,60\right\}$ with a  common maturity equal to the end of the trading horizon (6 months). The fluctuation of each strategy  depends on the moneyness of the option. When $S_t>>K$,  $N(d_+(t,S_t))$ is   close to $1$ and movements in the call   price mimic those in the underlying equity index. As s result, the strategy is roughly constant over time as shown by the bottom path in Figure \ref{fig:BShedge_simulation}(a). In contrast, if the option used is deep out-of-the-money ($S_t<<K$), then  $N(d_+(t,S_t))$ is   close to $0$, meaning that the investor needs to hold many units of this call, whose per-unit price is almost zero in such a scenario, to gain sufficient  direct exposure to the   index $S$. Consequently, small movements can lead to very large changes in the holdings over time (see the top path in Figure \ref{fig:BShedge_simulation}).

For the futures, even though the contract value is stochastic,  the strategy is time-deterministic with position becomes increasingly long  exponentially  over time at the risk-free rate. In  Figure \ref{fig:BShedge_simulation}(a), we compare the positions corresponding to the futures contracts with different maturities, and notice  that tracking with a  shorter-term futures  requires more units of futures in the portfolio.

  Figure \ref{fig:BS_tracking_strats} displays the strategies for   ``inverse tracking"  portfolios with exposure coefficient  $\beta=-1$. As expected, short positions are used, but  the   option strategy is most (resp. least) stable when the option is most in (resp. out of) the money ($K=40$).   The futures strategy is no longer time-deterministic, though the position still shows an  increasing trend towards maturity.   In fact, the stochastic futures strategy is now proportional to $S^{-2}$, as seen in \eqref{eq:futs_strat} when $\beta =-1$. Comparing across maturities, the shortest-term (resp. longest-term) futures has the most (least) short position, but the position remains negative for all maturities.

\subsection{Heston Model}\label{sec:heston}
We now discuss the tracking problem under the  \cite{hestonModel}  model for the equity index.  Under the risk-neutral measure, the dynamics of the reference  index   and stochastic volatility factor are given by
\begin{equation*}
	\begin{aligned}
		dS_t&=rS_tdt+\sqrt{Y_t}S_tdB_{t}^{\Q,0}\\
		dY_t&=\widetilde{\kappa}(\widetilde{\theta}-Y_t)dt+\nu\big(\rho\sqrt{Y_t}dB_{t}^{\Q,0}+ \sqrt{1-\rho^2}\sqrt{Y_t}dB_{t}^{\Q,1}\big),
	\end{aligned}
\end{equation*}where $B_{t}^{\Q,0}$ and $B_{t}^{\Q,1}$ are  two independent SBMs and $\rho\in(-1,1)$ is the instantaneous correlation parameter. 
The stochastic volatility factor  $Y$  is not traded and is driven by a Cox-Ingersoll-Ross (CIR) process. If we assume the Feller condition\index{Feller condition} $2\widetilde{\kappa}\widetilde{\theta}\ge\nu^2$ (see \citet{feller1951two}) and $Y_0>0$, then $Y$ stays strictly positive at all times almost surely under the risk-neutral measure.

Under the Heston Model, the tracking condition \eqref{eq:track_cond} becomes
\begin{equation*}
	\begin{aligned}
		\alpha_t=r(1-\beta_t)-\widetilde{\kappa}\left(\frac{\widetilde{\theta}}{Y_t}-1\right)\eta_t , \qquad 0 \le t \le T.
	\end{aligned}
\end{equation*}
The portfolio is subject to the stochastic drift  $\alpha_t$ which does not vanish as  long as $\eta_t \neq 0$ and $\beta_t \neq 1$. Therefore, perfect tracking is not achievable.

%As a special case,   the volatility-neutral portfolio with $\eta_t=0$   the tracking condition $\alpha_t=r(1-\beta_t)$, for all $t\in[0,T]$, which is identical to the tracking condition \eqref{eq:bs_replicating}  under  the Black-Scholes Model.  

%For the drift to be $0$ at time $t$, the variance process must take the value
%\begin{equation*}
%\begin{aligned}
%Y_t=\frac{\widetilde{\kappa}\eta{\widetilde{\theta}}}{\widetilde{\kappa}\eta+r(1-\beta)}\neq\widetilde{\theta}.
%\end{aligned}
%\end{equation*}
%But of course, even if the process takes this value for a single time point $t$, the process almost surely moves away from this value after $dt$ time units so the condition no longer holds. Nonetheless, we continue to derive the strategies achieving constant $\beta$ and $\eta$, recalling along the way that this drift is present.

Let us set  the coefficients to be  constant, i.e. $\beta_t=\beta$ and $\eta_t=\eta$ for all $t\in[0,T]$.  In the Heston Model,  a portfolio generally needs at least $d+1=2$ derivatives to control risk exposure with respect to the two sources of randomness.  We solve for the index exposure $\beta$  and factor exposure $\eta$ from the $2\times2$ system:
\[\left( \begin{array}{c}
\beta\\
\eta\\ \end{array} \right)= \left( \begin{array}{cc}
D_t^{(1)} & D_t^{(2)}\\
E_{t}^{(1)}& E_{t}^{(2)}\\
\end{array}\right)
\left( \begin{array}{c}
w_t^{(1)} \\
w_t^{(2)}\end{array} \right).\]
The second superscript is suppressed on the factor elasticities since there is only one exogenous factor here. By a simple inversion, the portfolio weights are  
\begin{equation}\label{eq:oneFactorStrats}
	\begin{aligned}
		w_t^{(1)}=\frac{\beta E_t^{(2)}-\eta D_t^{(2)}}{D_t^{(1)}E_t^{(2)}-D_t^{(2)}E_t^{(1)}}, \qquad  
		w_t^{(2)}=\frac{-\beta E_t^{(1)}+\eta D_t^{(1)}}{D_t^{(1)}E_t^{(2)}-D_t^{(2)}E_t^{(1)}}.
	\end{aligned}
\end{equation}
Of course, this solution is only valid if
\begin{equation}\label{eq:oneFactorDet}
	\begin{aligned}
		{D_t^{(1)}E_t^{(2)}\neq D_t^{(2)}E_t^{(1)}}.
	\end{aligned}
\end{equation}
For instance, using two European call (or put) options  on $S$ with different strikes can lead to the   trading strategy that generates the desired exposure associated with the given coefficients $\beta$ and $\eta$. However, issues may arise when only futures are used, as we will discuss next.

\begin{example}\label{ex:heston_futures}
Using two futures on $S$, with maturities $T_f^{(k)}$ ($k=1,2$), the corresponding portfolio weights are   given by 
\begin{equation}\label{eq:oneFactorStratsFutures}
 u_t^{(1)}=\frac{\beta H_t^{(2)}-\eta G_t^{(2)}}{G_t^{(1)}H_t^{(2)}-G_t^{(2)}H_t^{(1)}},  \quad \text{and} \quad 
u_t^{(2)}=\frac{-\beta H_t^{(1)}+\eta G_t^{(1)}}{G_t^{(1)}H_t^{(2)}-G_t^{(2)}H_t^{(1)}},
 \end{equation}
provided that ${G_t^{(1)}H_t^{(2)}\neq G_t^{(2)}H_t^{(1)}}$.\footnote{Again, the second superscript is suppressed on the factor elasticities since there is only one exogenous factor here.}  	Since the futures  prices are 	$f_t^{T_k}=S_te^{r(T_f^{(k)}-t)}$, for $k=1,2$,  the elasticities (see \eqref{EQG}-\eqref{EQH}) simplify to 
	\begin{equation*}
		\begin{aligned}
			G_t^{(k)}&=\frac{S_t}{f_t^{T_f^{(k)}}}\frac{\partial f^{T_f^{(k)}}}{\partial S}=1, \quad  \text{and} \quad 
			H_{t}^{(k)}&=\frac{Y_t}{f_t^{T_f^{(k)}}}\frac{\partial f^{T_f^{(k)}}}{\partial Y}=0.
		\end{aligned}
	\end{equation*}
However, this means that  ${G_t^{(1)}H_t^{(2)} = G_t^{(2)}H_t^{(1)}}=0$, so the strategy in  \eqref{eq:oneFactorStratsFutures} is not well defined. Hence,  it is generally  impossible to construct a futures portfolio that generates the desired exposure with respect to both the index and     stochastic volatility factor for any non-zero coefficients $\beta$ and $\eta$.
\end{example}

As shown in Example \ref{ex:heston_futures}, in order to gain exposure to $S$ and $Y$, the derivative need to have a non-zero sensitivity with respect to $Y$. If the investor does not seek exposure to $Y$  (i.e. $\eta=0$), then  she  only needs a single futures on $S$ to obtain the corresponding volatility-neutral portfolio. Next, we    show that by including   a futures on $Y$ we obtain a tracking portfolio  that can generate any desired exposure to $S$ and $Y$.

%\begin{equation*}
%	\begin{aligned}
%		1\cdot \frac{Y_t}{c_t}\frac{\partial c}{\partial Y}\neq \frac{S_t}{c_t}\frac{\partial c}{\partial S}\cdot 0\iff \frac{\partial c}{\partial Y}\neq 0.
%	\end{aligned}
%\end{equation*}
%Of course, one such example would be a call option on $S$. 

%Notice here that it is possible to get exposure to the stochastic volatility even if $\rho=0$. Put another way, even if the equity index is driven by a SBM that is uncorrelated with the SBM driving its volatility, one can control the exposure of the portfolio to the volatility. We display the strategies one might use to control these exposures in the following example. In the interest of explicit formulas, the example assumes the existence of a futures contract on the variance process, $Y$. 

%\begin{equation}
%	\begin{aligned}
%		c_t\equiv c(t,S_t,Y_t^{(1)},...,Y_t^{(d)})=\E^\Q\left[e^{-r(T-t)}h(S_T,Y_T^{(1)},...,Y_T^{(d)})\big|S_t,Y_t^{(1)},...,Y_t^{(d)}\right],
%	\end{aligned}
%\end{equation}

% By contrast to the previous example, the following example will demonstrate that other contracts can help us to attain path properties that are desirable. Namely, we can specify exposures to the index \emph{and} stochastic volatility. The previous example tells us that (at the very least) one of the derivatives must depend on the spot volatility. By employing a futures on the underlying variance process, one can specify the desired exposures.

\begin{example}\label{ex:heston_futures_plusVol}
	The Heston model can be  viewed as a joint   model for the market index   and  volatility index. The CIR process has also been used to model the volatility index  due to their common mean-reverting property (see  \cite{Grunbichler1996985} and \cite{futuresVIX}, among others).  Suppose there exist  a futures on the market index $S$ as well as a futures on the volatility index $Y$, and consider a dynamic portfolio of these  two futures contracts. We use the superscript $1$ to indicate the futures on the index (of maturity $T_f$) and the superscript $2$ to indicate the futures on the variance process (of maturity $T_y$). The prices of the   $T_f$-futures on $S$ and the $T_y$-futures on $Y$ are respectively given by 
		\begin{equation*}
	\begin{aligned}
 f_t^{T_f}=S_te^{r(T_f-t)} \quad \text{ and } \quad 	g^{T_y}_t =Y_te^{-\widetilde{\kappa}(T_y-t)}+\widetilde{\theta}(1-e^{-\widetilde{\kappa}(T_y-t)}).
	\end{aligned}
	\end{equation*}
	The relevant price elasticities are 
	\begin{equation*}
	\begin{aligned}
	G_t^{(1)}=1,\quad H_t^{(1)}=0,\quad G_t^{(2)} = 0, \quad \text{and} \quad H_{t}^{(2)}=\frac{Y_t}{g^{T_y}_t}e^{-\widetilde{\kappa}(T_y-t)}.
%	G_t^{(2)}&=\frac{S_t}{g^{T_y}_t}\frac{\partial g^{T_y}}{\partial S}=\frac{S_t}{g^{T_y}}\cdot 0 = 0,\\
%	H_{t}^{(2)}&=\frac{Y_t}{g^{T_y}_t}\frac{\partial g^{T_y}}{\partial Y}=\frac{Y_t}{g^{T_y}_t}e^{-\widetilde{\kappa}(T_y-t)}.
	\end{aligned}
	\end{equation*}
	The strategy $(u_t^{(1)},u_t^{(2)})$ achieving the exposure with coefficients $\beta$ and $\eta$ is found from the system:
	\[\left( \begin{array}{c}
	\beta\\
	\eta\\ \end{array} \right)= \left( \begin{array}{cc}
	G_t^{(1)} & G_t^{(2)}\\
	H_{t}^{(1)}& H_{t}^{(2)}\\
	\end{array}\right)
	\left( \begin{array}{c}
	u_t^{(1)} \\
	u_t^{(2)}\end{array} \right).\]
The system admits a unique solution, yielding the portfolio weights:
	\begin{equation*}
		\begin{aligned}
		u_t^{(1)}=\beta,\quad \text{and}\quad u_t^{(2)}=\eta+\eta\frac{\widetilde{\theta}}{Y_t}\left(e^{\widetilde{\kappa}(T_y-t)}-1\right).
		\end{aligned}
	\end{equation*}	
	Interestingly, the portfolio weight $u_t^{(1)}$ (resp. $u_t^{(2)}$) depends only on $\beta$ (resp. $\eta$).
\end{example}

Finally, we discuss the slippage process under the Heston model. Applying Proposition \ref{prop:pf_dynamics},  we obtain
	\begin{equation}\label{HestonZ}
	Z_t=r-r\beta-\widetilde{\kappa}\left(\frac{\widetilde{\theta}}{Y_t}-1\right)\eta+\frac{1}{2}\beta(1-\beta)Y_t+\frac{1}{2}\eta(1-\eta)\frac{\nu^2}{Y_t}-\beta\eta\nu\rho.
	\end{equation}
The term $\frac{1}{2}\beta(1-\beta)Y_t$ in \eqref{HestonZ} indicates  that the current slippage $Z_t$ depends on the instantaneous  variance $Y$ of the index $S$.  Similarly,   the term $\frac{1}{2}\eta(1-\eta)\frac{\nu^2}{Y_t}$ reflects  the dependence of the slippage on the instantaneous variance of the stochastic volatility $Y$. As in the Black-Scholes case, for   $\eta\notin[0,1]$, this term is negative. 	The final term $-\beta\eta\nu\rho$ is the  instantaneous  covariance between the index and the stochastic volatility. Since $\nu>0$, the term is positive whenever $\beta\eta\rho<0$. This happens either when (i) all three are negative, or (ii) exactly 1 is negative.  Since equity returns and volatility are typically negatively correlated\footnote{This phenomenon is called \emph{asymmetric volatility} and was first observed by \cite{blackAsymmetricVol}.}  ($\rho<0$), so going  long on both the index and  stochastic volatility ($\beta, \eta>0$) can  generate positive returns. This covariance term can  offset  some losses due to volatility decay.

\section{Volatility Index Tracking}\label{sec:volatility_idx}
In this section, we discuss tracking of the    volatility index, VIX, under two   continuous-time models.  These models were first applied to pricing volatility futures and options   by \cite{Grunbichler1996985} and \cite{futuresVIX}, among others. We expand their analysis to understand the tracking performance of VIX  derivatives portfolios and VIX ETFs.

\subsection{CIR Model}\label{sec:one_fact_CIR}
The volatility index, denoted by  $S$ in this section,  follows the CIR process:\footnote{Or  the square root  (SQR) process in the terminology of \cite{Grunbichler1996985}.} 
 \begin{align}\label{CIRQ}
dS_t = \widetilde{\kappa}\left(\widetilde{\theta}-S_t\right)dt + \sigma\sqrt{S_t} dB^\Q_t,
\end{align}
with constant parameters  $\widetilde{\kappa}, \widetilde{\theta},$ and  $\sigma >0$. If we assume the Feller condition\index{Feller condition} $2\widetilde{\kappa}\widetilde{\theta}\ge\sigma^2$ (see \citet{feller1951two}) and $S_0>0$, then $S$ stays strictly positive at all times almost surely. We omit the second  superscript on the SBM $B_t^\Q$ since there is only one SBM in this model. 

Under the CIR model, the tracking condition from \eqref{eq:track_cond} becomes 
\begin{equation}
\begin{aligned}
\alpha_t=r-\beta_t\widetilde{\kappa}\left(\frac{\widetilde{\theta}}{S_t}-1\right),
\end{aligned}
\end{equation}
for all $t\in[0,T]$. Therefore, for any given exposure coefficient $\beta_t$, there is a non-zero stochastic drift depending on the inverse of $S_t$.  In particular, if  $\beta_t=0$, then $\alpha_t=r$  and we recover the risk free rate by eliminating  exposure to the volatility  index. Next, suppose $\beta_t=1$ for a $100\%$ exposure to the volatility index. Then, the stochastic drift takes the form
\begin{equation}\label{eq:single_CIR_beta1}
\begin{aligned}
\alpha_t= \frac{\left(\widetilde{\kappa}+r\right)S_t-\widetilde{\kappa}\widetilde{\theta}}{S_t}, 
\end{aligned}
\end{equation}
which is negative (resp. positive) whenever $S_t$ is iss below (resp. above)  the critical level $\frac{\widetilde{\kappa}\widetilde{\theta}}{\widetilde{\kappa}+r}$. In addition, if $r=0$, then  the critical value is equal to $\widetilde{\theta}$.\footnote{Alternatively, we can assume $\widetilde{\kappa}>>r$ so that $\frac{\widetilde{\kappa}}{\widetilde{\kappa}+r}\approx1$ and the critical value is approximately $\widetilde{\theta}$.} Then, as $S_t$ mean-reverts to $\widetilde{\theta}$, the drift is on average equal to $0$, but is stochastic nonetheless. 

 Furthermore, according to Proposition \ref{prop:pf_dynamics}, the slippage process given any constant  $\beta \in \R$  in the CIR Model is given by
	\begin{equation*}
		\begin{aligned}
			Z_t &=r-\beta\widetilde{\kappa}\left(\frac{\widetilde{\theta}}{S_t}-1\right)+\frac{1}{2}\beta(1-\beta)\frac{\sigma^2}{S_t}.
		\end{aligned}
	\end{equation*}
	The second term reflects the mean reverting path behavior of $S$. And since $\widetilde{\theta}$ is the long-run mean of $S_t$,   this term is expected to stay around zero over time, though the deviation from zero is proportional to $\beta\widetilde{\kappa}$. The last term involves  the   variance of the volatility index, $\frac{\sigma^2}{S_t}$, and    is strictly negative for $\beta\notin[0,1]$, leading to value erosion.  Lastly, we observe that  $Z_t$ is  an affine function of  ${S_t}^{-1}$, which is an    inverse CIR process.  The moments and other statistics of such a process   are well known (see \cite{ahnGao1999}), so this form will be useful for understanding the distribution and computing expectations of $Z_t$.

Since there are $d=0$ factors outside of the volatility  index, $d+1=1$ derivatives allow for a unique tracking strategy. Specifically, we study the dynamics of portfolios of   futures    written on $S$.  For any   maturity $T$, the futures price is  
\begin{align}\label{fTCIR}
f^T_t := f(t,S,T) = \E^{\Q}\left[S_T|S_t=S\right] = (S-\widetilde{\theta})e^{-\widetilde{\kappa}(T-t)} +\widetilde{\theta}.
\end{align}
 In Figure \ref{fig:termstructure} we display two  term structures of VIX futures on two different dates, calibrated to the CIR model. The term structure can be either increasing   concave or decreasing   convex (see e.g. \cite{LiVIX} and Chap 5 of \cite{meanReversionBook}.) While the good fits  further suggest that  the CIR Model is a suitable model for  VIX, we remark that there exist examples of  irregularly shaped VIX term structures and calibrated  model parameters often change over time. This motivates us to investigate the tracking problem under a more sophisticated VIX model in the next section.

%Without loss of generality, we solve the index exposure equation for the strategy weight:
%\begin{equation*}
%\begin{aligned}
%&w_t D_t=\beta\implies&w_t=\frac{\beta c_t}{S_t}\left(\frac{\partial c}{\partial S}\right)^{-1}.
%\end{aligned}
%\end{equation*}
%

\begin{figure}[th]
	\begin{centering}
		\subfigure[Feb. 23, 2009]{\includegraphics[trim={1.2cm 0.5cm 1.5cm 0.5cm},clip,width=3in]{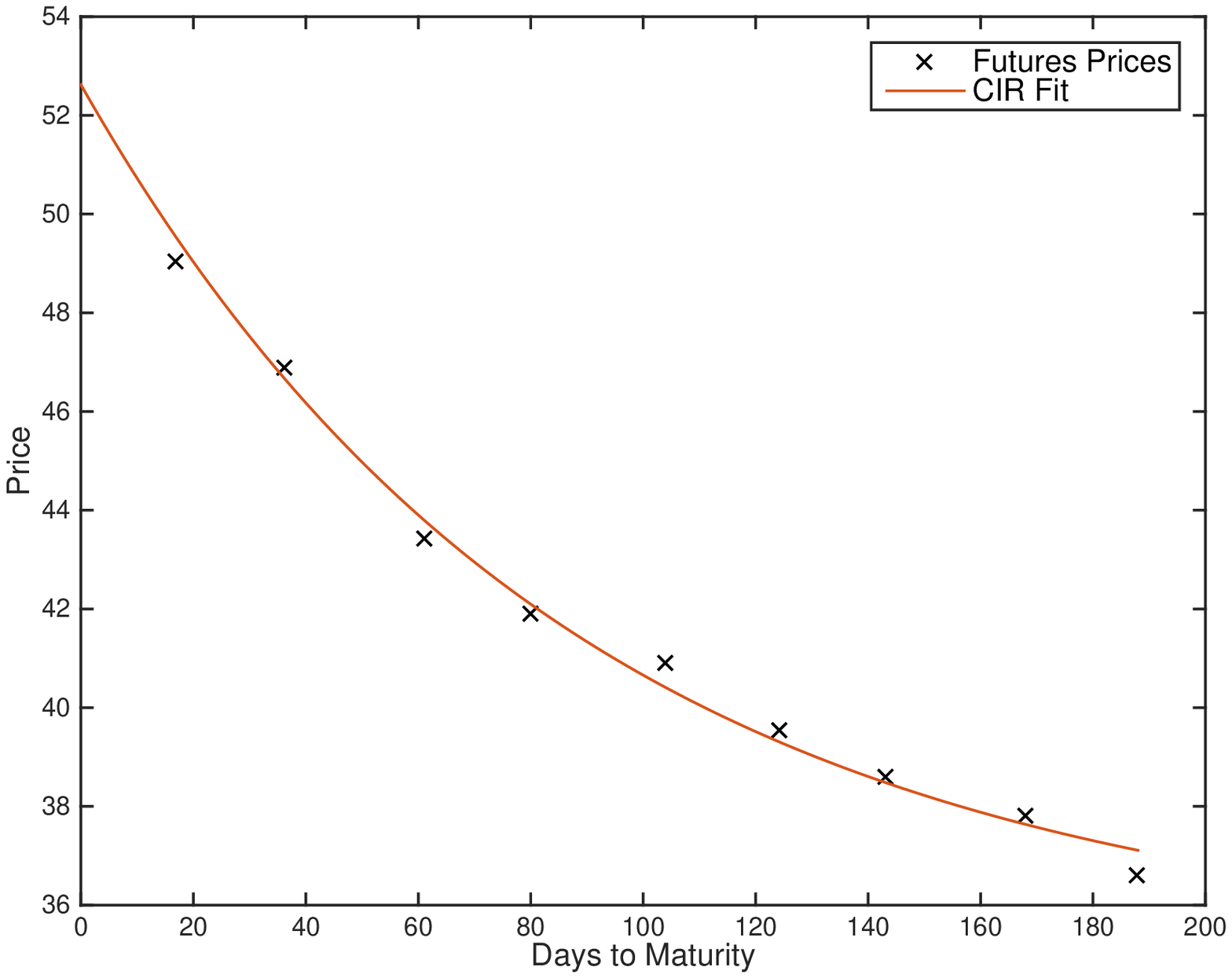}}
		\subfigure[Dec. 11, 2009]{\includegraphics[trim={1.2cm 0.5cm 1.5cm 0.5cm},clip,width=3in]{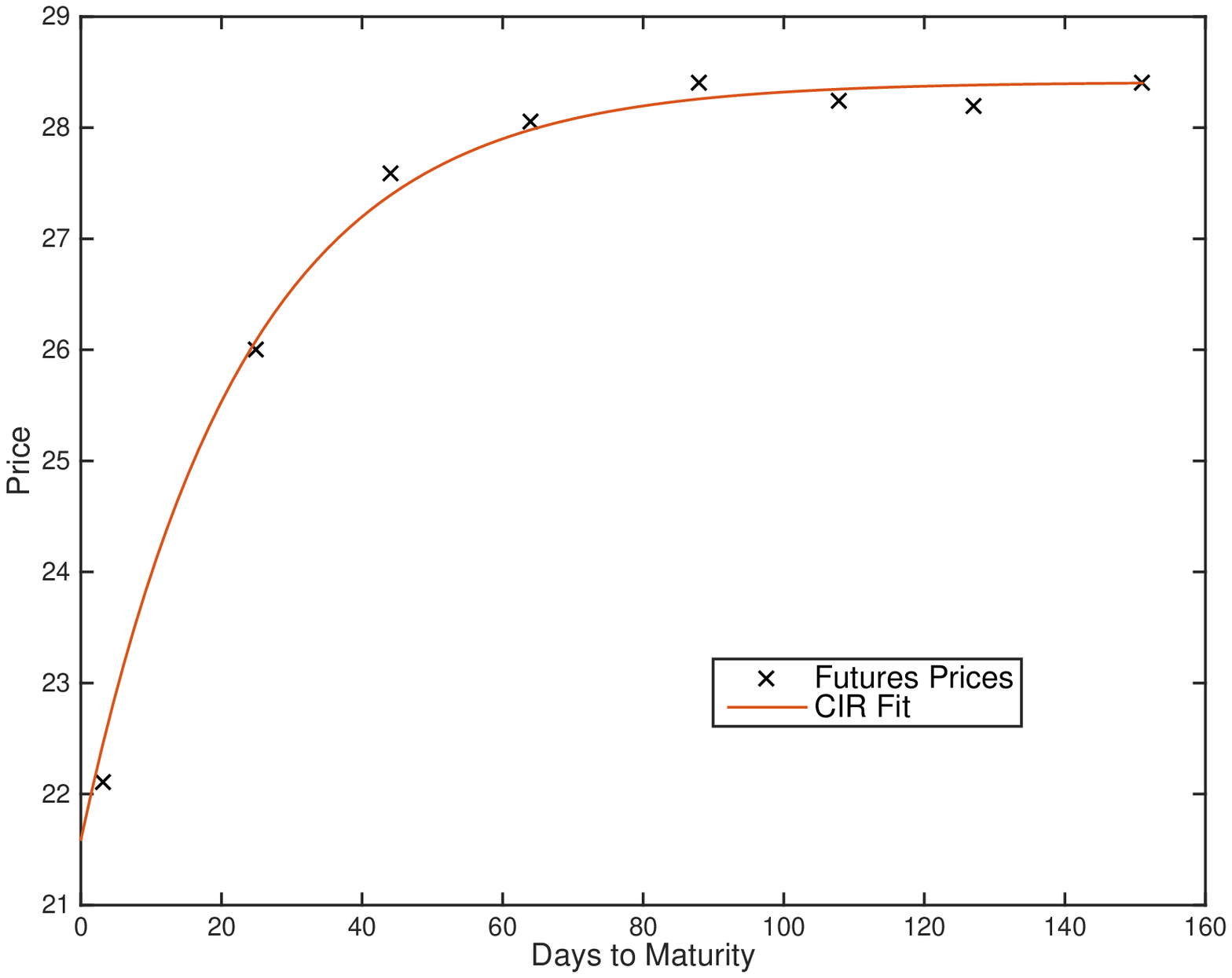}}\caption{\small{Term structure of VIX futures as observed on (a) Feb. 23, 2009 and (b) Dec. 11, 2009. The term structure has changed  from  decreasing convex to increasing  concave.}}
		\label{fig:termstructure}
	\end{centering}
\end{figure}
%
%For tracking the index, the relevant elasticity is
%\begin{equation*}
%\begin{aligned}
%G_t=\frac{S_t}{f_t^{T}}e^{-\widetilde{\kappa}(T-t)}.
%\end{aligned}
%\end{equation*}
%Thus, we have

The futures trading strategy  is given by 
\begin{equation} \label{utvix}
\begin{aligned}
u_t =\beta+\frac{\beta\widetilde{\theta}}{S_t}\left(e^{\widetilde{\kappa}(T-t)}-1\right).
\end{aligned}
\end{equation} 
Of course, one may set $\beta=1$ to seek direct exposure to the volatility index. A number of VIX ETFs/ETNs attempt to gain direct exposure to VIX (e.g. VXX) by constructing a futures portfolios with time-deterministic weights. The following example elucidates how and to what extent such an  exchange-traded product falls short of this goal.

			\begin{figure}[H]
		\begin{centering}
			\includegraphics[trim={1.2cm 0.5cm 1.5cm 0.3cm},clip,width=5in]{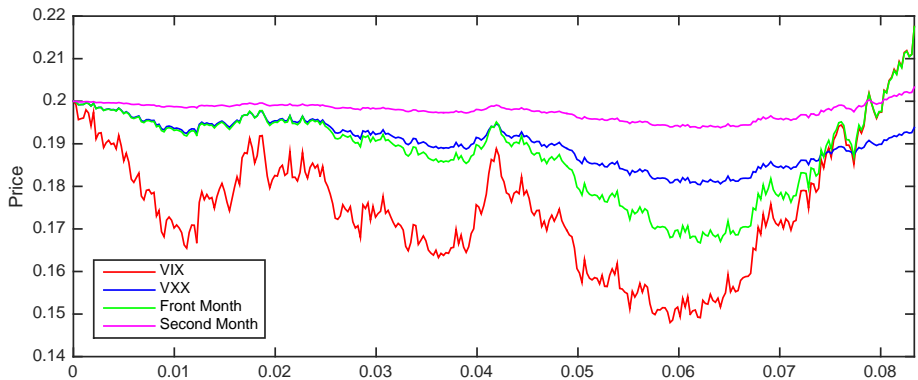}
			%\subfigure[Two Months]{\includegraphics[trim={1.2cm 0.5cm 1.5cm 0.3cm},clip,width=5in]{vxx_simulation_v3}}
			\caption{\small{Simulation of VXX along with VIX and the two futures contracts VXX holds over  1 month ($T=1/12$). Parameters are $S_0=0.2$, $\widetilde{\kappa}=20$, $\widetilde{\theta}=0.2$ and $\nu=0.4$. }}
 			\label{vxx_simulation}
		\end{centering}
	\end{figure}
	
\subsection{Comparison to VXX}\label{sec:VXXcomp}

Let us consider a portfolio  of futures with  a time-deterministic strategy. Its value evolves according to 
	\begin{align}\label{Vwtd}
	\frac{dV_t}{V_t} = u(t)\frac{df^{T_{i(t)}}_t}{f^{T_{i(t)}}_t} + (1-u(t))\frac{df^{T_{i(t)+1}}_t}{f^{T_{i(t)+1}}_t}+rdt,
	\end{align}
	where $i(t):=\min\{i: T_{i-1}< t \leq T_{i}\}$, $T_0:=0$, and the portfolio weight for the $T_i$-futures is 
	\begin{align}\label{wtd}
	u(t) := \frac{T_{i(t)}-t}{T_{i(t)} - T_{i(t)-1}}.
	\end{align}
This is the strategy employed by the popular VIX ETN, iPath S\&P 500 VIX Short-Term Futures ETN (VXX). The strategy starts by investing 100\% in the front-month VIX futures contract, and decreases its holding   linearly from 100\% to 0\% while  the weight on the second-month contract increases linearly from 0\% to 100\%.

  Figure \ref{vxx_simulation} illustrates such a portfolio with simulated index prices. Specifically,  we plot   the VIX   in red, and the VXX price in blue over one month. The component futures (front and second month) are plotted in green and purple (respectively). Compared to the VIX, the   VXX  is significantly  less volatile. This can be confirmed analytically since 
	\begin{align*}
	\left(\frac{dV_t}{V_t}\right)^2 = \left(\frac{u(t)}{S_t+\widetilde{\theta} (e^{\widetilde{\kappa}(T_{i(t)}-t)}-1)} + \frac{1-u(t)}{S_t+\widetilde{\theta} (e^{\widetilde{\kappa}(T_{i(t)+1}-t)}-1)}  \right)^2 (dS_t)^2 < \left(\frac{dS_t}{S_t}\right)^2,
	\end{align*}
	where the inequality is due to $\widetilde{\kappa}(T_j-t)>0$ so that $\widetilde{\theta}\left(e^{\widetilde{\kappa}(T_j-t)}-1\right)>0$ for any $j$.	As seen in Figure \ref{fig:VXX_returns_evolution}(a), both  VXX and the dynamic portfolio cannot perfectly track VIX over time. Nevertheless,  compared to VXX, the dynamic portfolio is more reactive to changes in VIX. In Figure \ref{fig:VXX_returns_evolution}(b), we plot two scatterplots of the annualized  returns  between the dynamic portfolio and VIX (top) and between VXX and VIX (bottom). On both plots, a solid straight line with slope 1 is a drawn for comparison. The dynamic portfolio, while not   tracking VIX perfectly, generates highly similar returns as VIX, and is visibly much closer to VIX than VXX. 
	
	\newpage 
	\white{.}	\vspace{20pt}
		\begin{figure}[H]
		\begin{centering}
			\subfigure[Portfolio Evolution]{\includegraphics[trim={1.2cm 0.5cm 1.5cm 0.3cm},clip,width=5in]{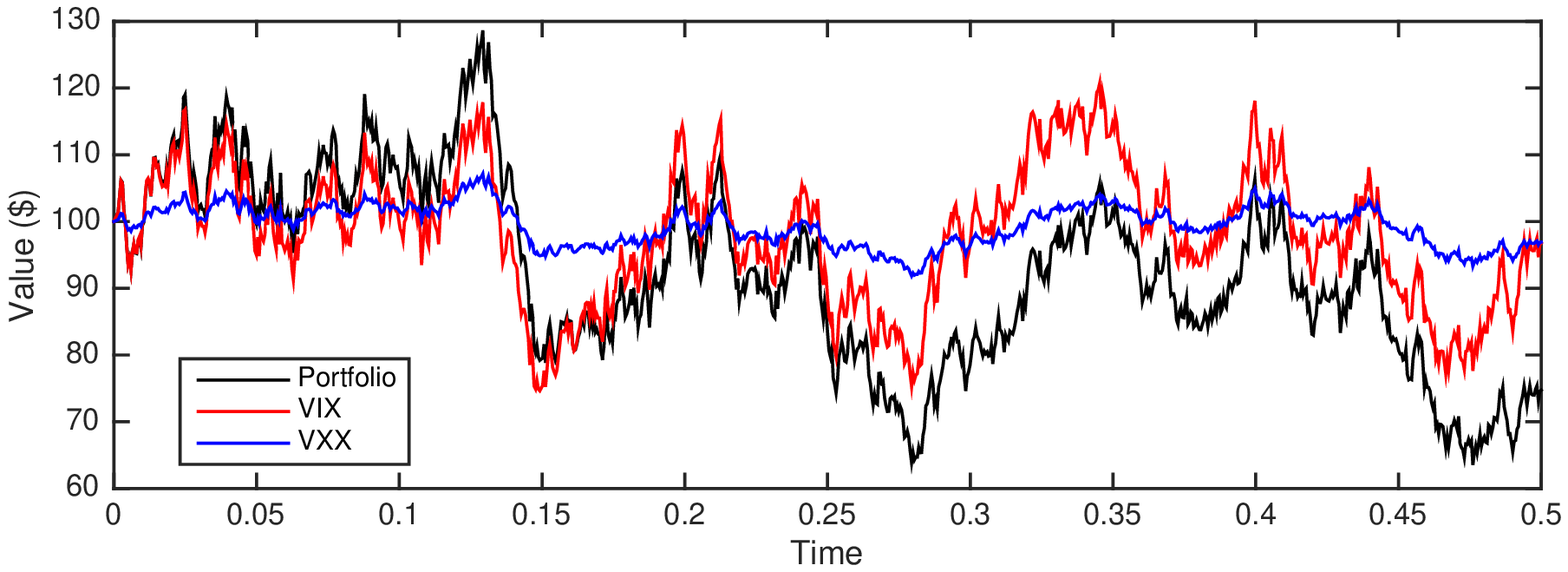}}			\subfigure[Return Scatter Plots]{\includegraphics[trim={1.2cm 0.5cm 1.5cm 0.3cm},clip,width=5in]{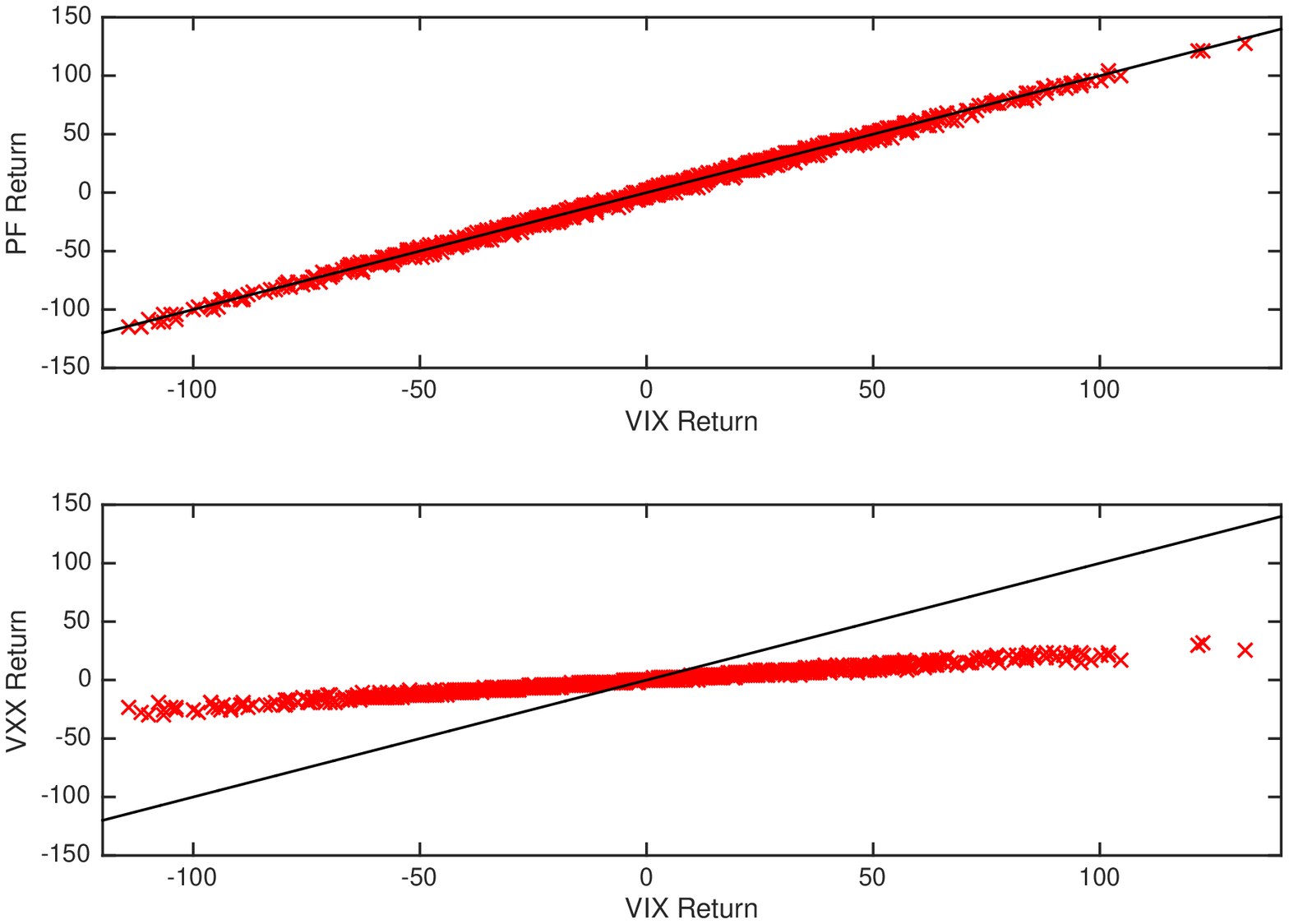}}
			\caption{\small{Sample paths of VIX, VXX,  and the dynamic portfolio with $\beta=1$ over a 6-month period ($T=0.5$). Parameters are $S_0=0.2$, $\widetilde{\kappa}=20$, $\widetilde{\theta}=0.2$, and $\nu=0.4$. Panel (a) displays the time series of the portfolio value, along with the price paths of VIX and VXX, over time starting at  \$$100$. Panel (b) shows  the scatter plots of the annualized  returns  between the portfolio and VIX (top) and between VXX and VIX (bottom). On both plots, a solid straight line with slope 1 is a drawn for visual comparison. }}
			\label{fig:VXX_returns_evolution}
		\end{centering}
	\end{figure}
	
	\clearpage
	Given the VXX strategy (or any other strategy), we can infer from the SDE of the portfolio (see \eqref{Vwtd})  the corresponding  drift (the $dt$ term) and   exposure to $S$  (coefficient of $dS_t/S_t$), which we denote by  $\alpha_t^{(V)}$ and  $\beta_t^{(V)}$, respectively. As a point of reference, if   a portfolio tracks VIX  one-to-one perfectly, then we have $\beta_t^{(V)}= 1$ and $\alpha_t^{(V)}= 0$. As our earlier  discussion following equation \eqref{eq:single_CIR_beta1} indicates, this perfect tracking is impossible, but we still wish to illustrate how much VXX deviates from   VIX in terms of the implied values of $\alpha_t^{(V)}$ and  $\beta_t^{(V)}$.
	
Let us denote   $T_1$ and $T_2$  as the maturities for the    front-month and second-month futures, respectively.  With the portfolio weights in \eqref{wtd}, we find from the portfolio's SDE \eqref{Vwtd}   that
	\begin{equation*}
	\begin{aligned}
		\alpha_t^{(V)}&=r+\beta_t^{(V)}\widetilde{\kappa}\left(1-\frac{\widetilde{\theta}}{S_t}\right),\\
	\beta_t^{(V)}&=\frac{S_t}{f_t^{T_1}}e^{-\widetilde{\kappa}(T_1-t)}-\frac{tS_te^{\widetilde{\kappa}t}}{T_1}\left(\frac{e^{-\widetilde{\kappa}T_1}}{f_t^{T_1}}-\frac{e^{-\widetilde{\kappa}T_2}}{f_t^{T_2}}\right).
	\end{aligned}
	\end{equation*}
	As expected, we do not have $\beta_t^{(V)}= 1$ and $\alpha_t^{(V)}= 0$. Indeed both coefficients are stochastic and depend on the level of $S_t$.	In Figure \ref{fig:VXX_returns_evolution}, we illustrate the sample paths of VIX, VXX, and the dynamic   portfolio of two futures with exposure coefficient $\beta=1$. The last portfolio attempts to  track VIX one-to-one but must be  subjected to a stochastic drift $\alpha_t\neq0$.

		In Figure \ref{fig:VXX_beta_alpha}, we plot the implied $\alpha$ and implied  $\beta$ over time based on the sample path  of VXX in Figure \ref{fig:VXX_returns_evolution}.  The implied  $\beta$  for VXX varies significantly  over time, and is far from the reference value $1$ (for unit exposure to VIX) all the time. On average the implied $\beta$ fluctuates  around $0.23$ over the entire 6 month period. Recall that VXX is a long portfolio of the front-month and second-month futures, starting with a 100\% allocation in the front-month   at each maturity. Interestingly, as VXX allocates more in the second-month futures over time, the portfolio value reaches its maximum approximately halfway through each contract period, suggesting that concentrating on the shortest-term futures does not imply the closest tracking to the underlying index. In Figure \ref{fig:VXX_beta_alpha}(b), we plot the stochastic   $\alpha$ in units of annualized basis points for both VXX and the dynamic portfolio with $\beta = 1$. For both portfolios,  $\alpha$ is not 0 as expected. For the dynamic portfolio with $\beta=1$, the stochastic drift is more volatile than that for VXX, but both are  relatively small compared to the returns seen in Figure \ref{fig:VXX_returns_evolution}(b). 
		
			We conclude this example by discussing the futures trading strategies  corresponding to   $\beta=1$ (see \eqref{utvix}) with \emph{either} the front-month futures only, \emph{or} the second-month futures only. Recall that in this  model  only $d+1=1$ derivative product is necessary to achieve unit exposure ($\beta=1$), but the choice of derivative can lead to a very different portfolio weight over time. To see this, we plot in  Figure \ref{fig:VXX_beta_pf_weights} the  sample paths of the two portfolio weights corresponding to  the two futures contracts with different maturities.   	In general, the futures  strategies tend to decrease exponentially in each maturity cycle,   and become discontinuous at maturities as the portfolio rolls into the new futures contract. The  front-month futures strategy  decays roughly  from  5 to   1  in each cycle. Given  that futures price will converge  to index price   by maturity, it is intuitive that the strategy weight becomes 1 for front-month futures. Intuitively,  futures prices  are typically less volatile than the index, so  leveraging (weights greater than 1) is expected and can in effect  increase the portfolio's volatility to attempt to better track the index. Comparing between the two strategies, using  the second-month futures to track VIX leads to significant leveraging.

			One of the advantages of futures contracts is the ease of leveraging due to margin requirements. If margin requirements are 20\% of notional, then it is possible to achieve $5\times$ leverage. Margin requirements are constantly changing for VIX futures.\footnote{See \url{http://cfe.cboe.com/margins/CurDoc/Default.aspx} for current margin requirements for VIX futures. Margin requirements are stated as a dollar value, rather than a percentage. That dollar value is based on one unit of VIX futures, which is 1,000 times the stated futures price. Back-of-the-envelope calculations yield historical percentages between 10\% and 30\%.} However, $25\times$ to $30\times$ for the second-month futures is neither typical nor    practical. The weights for front-month futures are more in line with feasible leverage that one can attain. Recall that the strategy employed by VXX is a time-deterministic one. Figure  \ref{fig:VXX_beta_pf_weights} further illustrates  that stochastic portfolio weights are necessary in order to achieve unit exposure to the volatility index.

%	\newpage
%	\white{.}
%	\vspace{20pt}
%		\begin{figure}[H]
%		\begin{centering}
%			\subfigure[Portfolio Evolution]{\includegraphics[trim={1.2cm 0.5cm 1.5cm 0.3cm},clip,width=5in]{PF_evolutionSingleCIR_v3}}			\subfigure[Return Scatterplots]{\includegraphics[trim={1.2cm 0.5cm 1.5cm 0.3cm},clip,width=5in]{returnScatterSingleCIR_v3}}
%			\caption{\small{Simulation of VXX along with VIX and the dynamic portfolio with $\beta=1$ over a 6-month period ($T=0.5$). Parameters are $S_0=0.2$, $\widetilde{\kappa}=20$, $\widetilde{\theta}=0.2$ and $\nu=0.4$. Panel (a) displays the time series of the portfolio value, along with the price paths of VIX and VXX, over time starting at  \$$100$. Panel (b) shows  the scatter plots of the annualized  returns  between the portfolio and VIX (top) and between VXX and VIX (bottom). On both plots, a solid straight line with slope 1 is a drawn for comparison. }}
%			\label{fig:VXX_returns_evolution}
%		\end{centering}
%	\end{figure}
%	
%	\clearpage

%	\vspace{20pt}

	\begin{figure}[H]
		\begin{centering}
			\subfigure[Implied $\beta_t$]{\includegraphics[trim={1.2cm 0.5cm 1.5cm 0.3cm},clip,width=5in]{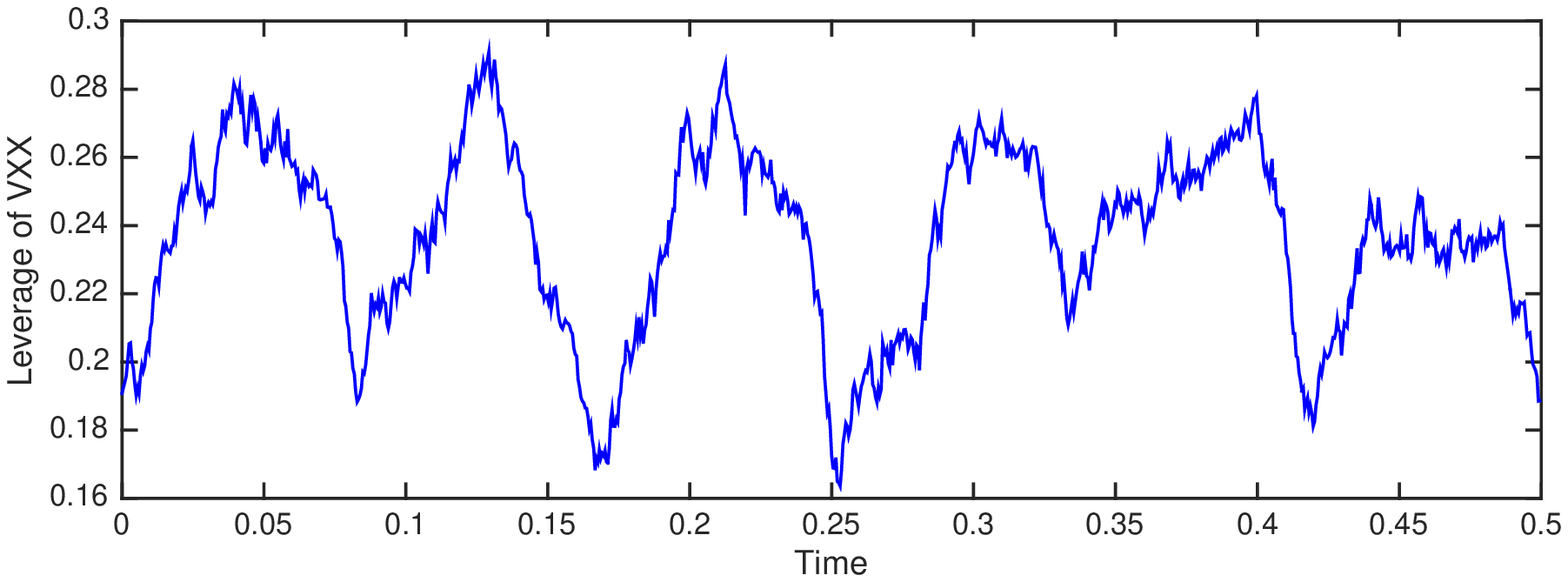}}
			\subfigure[Implied $\alpha_t$]{\includegraphics[trim={1.2cm 0.5cm 1.5cm 0.3cm},clip,width=5in]{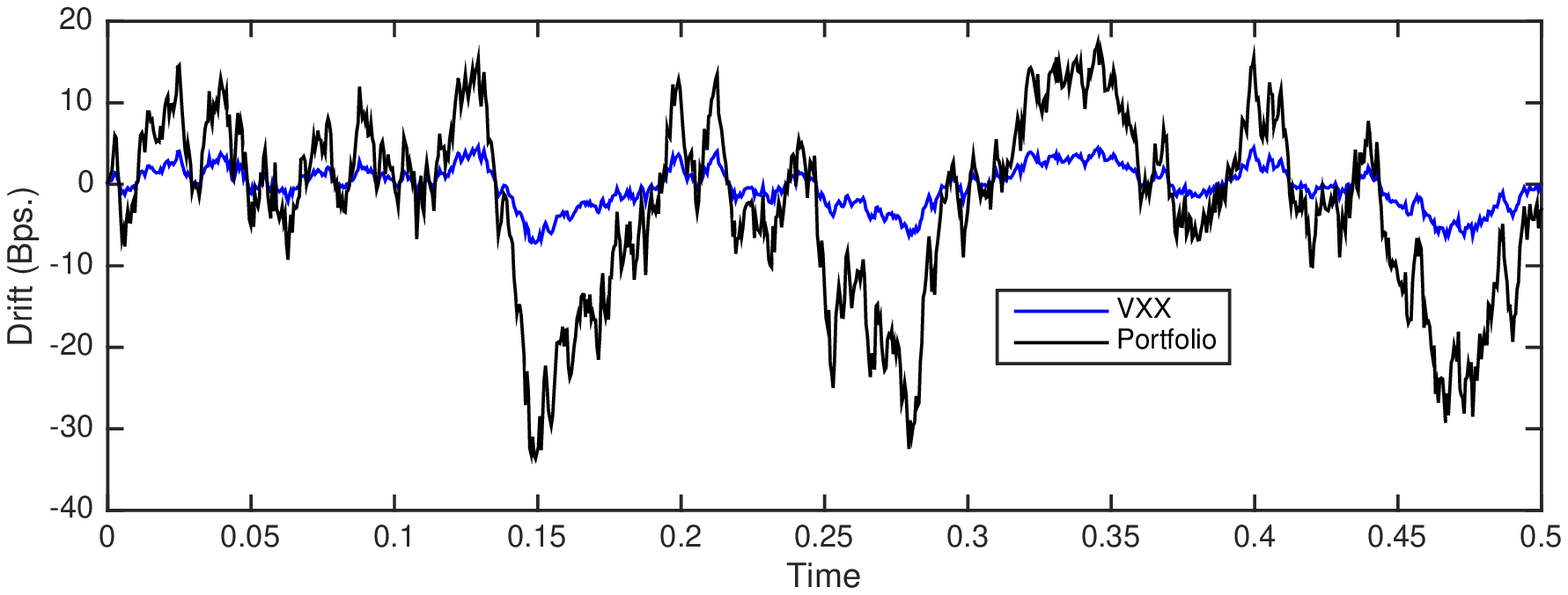}}
			\caption{\small{With reference to the  simulated paths in Figure \ref{fig:VXX_returns_evolution}, we show   (a) the implied exposure coefficient $\beta_t$ inferred from VXX , and (b)   stochastic drift $\alpha_t$ for both VXX and the tracking  portfolio  over a six-month period ($T=0.5$). }}
			\label{fig:VXX_beta_alpha}
		\end{centering}
	\end{figure}

\vspace{20pt}

	\begin{figure}[H]
		\begin{centering}
			\includegraphics[trim={1.2cm 0.5cm 1.5cm 0.3cm},clip,width=5in]{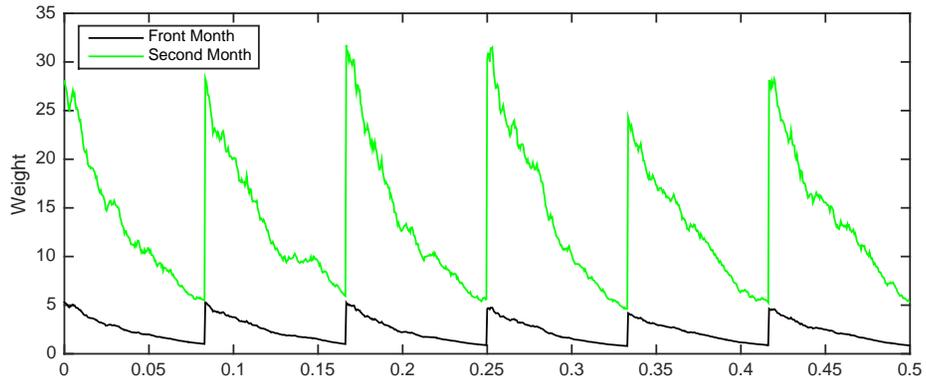}
			\caption{\small{Portfolio  weights  of the  dynamic portfolio with exposure coefficient  $\beta=1$  in Figure \ref{fig:VXX_returns_evolution}.}}
			\label{fig:VXX_beta_pf_weights}
		\end{centering}
	\end{figure}

	  \clearpage
	   
 \subsection{CSQR Model}\label{sec:two_fact_CIR}
We now investigate the tracking strategy and slippage process under an extension of the CIR Model. In this section, $S$ will continue to be mean reverting, but the long-run mean   is also  stochastic. This model is referred to as  the concatenated square root process (CSQR) (see \cite{futuresVIX}):
\begin{equation*}
\begin{aligned}
	dS_t&=\widetilde{\gamma}(Y_t-S_t)dt+\sigma \sqrt{S_t}dB_{t}^{\Q,0},\\
	dY_t&=\widetilde{\kappa}(\widetilde{\theta}-Y_t)dt+\nu\rho\sqrt{Y_t}dB_{t}^{\Q,0}+\nu\sqrt{1-\rho^2}\sqrt{Y_t}dB_{t}^{\Q,1}.
\end{aligned}
\end{equation*}
Here, $B_{t}^{\Q,0}$ and $B_{t}^{\Q,1}$ are independent SBMs.  We assume the parameters are chosen so that the pair $(S_t,Y_t)$ is strictly positive $\forall\,t\ge0$.\footnote{See for example \cite{duffieKan} for a discussion of models of this form.} Here, the index tends to   revert  to the stochastic  level $Y_t$ which is also mean-reverting.  This accounts for  the empirical observations of the path behavior of VIX, and that VIX futures calibration suggest that the long-run mean  oscillates  over time (see e.g.  Figure \ref{fig:termstructure}). 

%For the CSQR Model, the tracking condition from  becomes
%\begin{equation*}
% 	\begin{aligned}
%\alpha_t-r+\frac{\widetilde{\gamma}(Y_t-S_t)}{S_t}\beta_t+\frac{\widetilde{\kappa}(\widetilde{\theta}-Y_t)}{Y_t}\eta_t=0,
% 	\end{aligned}
%\end{equation*}
%for all $t\in[0,T]$. The left hand side depends on the paths of $S$ and $Y$ as well as $\alpha$, $\beta$, $\eta$. 

Let the investor set the values of exposure coefficients $\beta_t \equiv\beta$ and $\eta_t\equiv \eta$, then the portfolio is subject to  a stochastic  drift  (see Proposition \ref{prop:rep_cond}) that is a function of $S_t$ and $Y_t$:
 \begin{equation*}
\begin{aligned}
\alpha_t=r-\frac{\widetilde{\gamma}(Y_t-S_t)}{S_t}\beta -\frac{\widetilde{\kappa}(\widetilde{\theta}-Y_t)}{Y_t}\eta.
\end{aligned}
\end{equation*}
Also, applying Proposition \ref{prop:pf_dynamics}, we obtain the slippage process   $Z_t$ in the CSQR  model:
	\begin{equation*}
	\begin{aligned}
	Z_t=r-\beta\widetilde{\gamma}\left(\frac{Y_t}{S_t}-1\right)-\eta\widetilde{\kappa}\left(\frac{\widetilde{\theta}}{Y_t}-1\right)+\frac{1}{2}\beta(1-\beta)\frac{\sigma^2}{S_t}+\frac{1}{2}\eta(1-\eta)\frac{\nu^2}{Y_t}-\frac{\beta\eta\nu\rho\sigma}{\sqrt{S_tY_t}}.
%	Z_t&=\alpha+\frac{1}{2}\beta(1-\beta)\frac{\sigma^2}{S_t}+\frac{1}{2}\eta(1-\eta)\frac{\nu^2}{Y_t}-\frac{\beta\eta\nu\rho\sigma}{\sqrt{S_tY_t}}\\
%	&=r-\beta\widetilde{\gamma}\left(\frac{Y_t}{S_t}-1\right)-\eta\widetilde{\kappa}\left(\frac{\widetilde{\theta}}{Y_t}-1\right)+\frac{1}{2}\beta(1-\beta)\frac{\sigma^2}{S_t}+\frac{1}{2}\eta(1-\eta)\frac{\nu^2}{Y_t}-\frac{\beta\eta\nu\rho\sigma}{\sqrt{S_tY_t}}.
	\end{aligned}
	\end{equation*}
As we can see, the first two  terms   reflect the mean-reverting properties of $S_t$ and $Y_t$. On average, their effects tend  to be zero,  but are also proportional to the speeds of mean reversion $\widetilde{\gamma}$ and $\widetilde{\kappa}$, and exposure coefficients $\beta$ and $\eta$. The next two terms account for  the   variances of  the index and its stochastic mean, and are negative if $\beta\notin[0,1]$ and $\eta\notin[0,1]$, respectively. The final term reflects the effect of covariance between $S$ and $Y$ in  this model. It will be negative if either (i) all three of $\beta,\eta$ and $\rho$ are positive, or (ii) exactly one of $\beta,\eta$, or $\rho$ is positive.

Since $d=1$, we know that $d+1=2$ derivative products are required to obtain the desired exposures. Let us consider the use of futures contracts on $S$. The   price of  a futures written on $S$  with maturity $T_k$ is given by
\[f_t^{T_k}= \widetilde{\theta}+\left(S_t-\widetilde{\theta}\right)e^{-\widetilde{\gamma}(T_k-t)}+\begin{cases} 
\widetilde{\gamma}e^{-\widetilde{\gamma}T_k}\left(Y_t-\widetilde{\theta}\right)e^{\widetilde{\kappa}t}(T_k-t), & \widetilde{\gamma}=\widetilde{\kappa} \\
\frac{\widetilde{\gamma}\left(Y_t-\widetilde{\theta}\right)}{\widetilde{\gamma}-\widetilde{\kappa}}\left(e^{-\widetilde{\kappa}(T_k-t)}-e^{-\widetilde{\gamma}(T_k-t)}\right), & \widetilde{\gamma}\neq\widetilde{\kappa},
\end{cases}
\] for $t \le T_k$; see \cite{futuresVIX} for a derivation. Notice the first term would be the futures price if $S_t$ were   reverting to the mean $\widetilde{\theta}$ at speed $\widetilde{\gamma}$. Next we calculate the sensitivities  with respect to the index and stochastic mean:
\[\frac{\partial f_t^{T_k}}{\partial S}=e^{-\widetilde{\gamma}(T_k-t)}, \quad \text{ and } \quad \frac{\partial f_t^{T_k}}{\partial Y}=\begin{cases} 
\widetilde{\gamma}e^{\widetilde{\kappa}t-\widetilde{\gamma}{T_k}}({T_k}-t), & \widetilde{\gamma}=\widetilde{\kappa} \\
\frac{\widetilde{\gamma}}{\widetilde{\gamma}-\widetilde{\kappa}}\left(e^{-\widetilde{\kappa}({T_k}-t)}-e^{-\widetilde{\gamma}({T_k}-t)}\right), & \widetilde{\gamma}\neq\widetilde{\kappa} 
\end{cases}.\]
  Henceforth, assume the latter case that $\widetilde{\gamma}\neq\widetilde{\kappa}$.\footnote{Similar proofs can be found in Appendices \ref{App2} and \ref{App3} for $\widetilde{\gamma}=\widetilde{\kappa}$.} In that case, the respective  elasticities are given by
\begin{align}\label{eq:csqr_spotE}
G_t^{(k)}=\frac{S_t}{f_t^{{T_k}}}e^{-\widetilde{\gamma}({T_k}-t)},\quad \text{and} \quad H_t^{(k)}=\frac{Y_t}{f_t^{{T_k}}}\frac{\widetilde{\gamma}}{\widetilde{\gamma}-\widetilde{\kappa}}\left(e^{-\widetilde{\kappa}({T_k}-t)}-e^{-\widetilde{\gamma}({T_k}-t)}\right).
\end{align}
Just as in Section \ref{sec:heston}, the second numerical superscript on $H_t^{(k)}$ has been suppressed since there is only one exogenous factor in this model. 

We consider a portfolio of two  futures contracts, both on $S$ with maturities $T_2>T_1\ge T$. Recall that tracking  does not work with two futures in the Heston Model, so one must include another  type of derivative (e.g. option). Now  in the CSQR model, we  check if this futures portfolio works by verifying 
\begin{equation}\label{eq:csqr_exist}
\begin{aligned}
{G_t^{(1)}H_t^{(2)}\neq G_t^{(2)}H_t^{(1)}},
\end{aligned}
\end{equation}
for all $t\in[0,T]$. Upon plugging in the above elasticities, we find that the condition is equivalent to $T_1\neq T_2$ (see Appendix \ref{App4}.) Since we have assumed that  $T_1<T_2$,   the condition holds. Therefore, the linear system of equations    
\begin{equation}\label{eq:line_sys_csqr}
\left( \begin{array}{c}
\beta\\
\eta\\ \end{array} \right)= \left( \begin{array}{cc}
G_t^{(1)} & G_t^{(2)}\\
H_{t}^{(1)}& H_{t}^{(2)}\\
\end{array}\right)
\left( \begin{array}{c}
u_t^{(1)} \\
u_t^{(2)}\end{array} \right),
\end{equation}
which yields the strategy for the two futures contracts is always solvable. 

Soling the system yield  the portfolio weights
\begin{equation}\label{u1eq}
 u_t^{(1)}=\frac{\beta f_t^{T_1}}{S_t}\left(\frac{ e^{\widetilde{\gamma}(T_2-t)}-e^{\widetilde{\kappa}(T_2-t)}}{e^{\widetilde{\gamma}(T_2-T_1)}-e^{\widetilde{\kappa}(T_2-T_1)}}\right)-\frac{\eta f_t^{T_1}}{Y_t}\left(1-\frac{\widetilde{\kappa}}{\widetilde{\gamma}}\right)\left(\frac{ e^{\widetilde{\kappa}(T_2-t)} }{e^{\widetilde{\gamma}(T_2-T_1)}-e^{\widetilde{\kappa}(T_2-T_1)}}\right),
\end{equation}
and
\begin{equation}\label{u2eq}
u_t^{(2)}=\frac{-\beta f_t^{T_2}}{S_t}\left(\frac{e^{\widetilde{\gamma}(T_1-t)}-e^{\widetilde{\kappa}(T_1-t)}}{e^{-\widetilde{\kappa}(T_2-T_1)}-e^{-\widetilde{\gamma}(T_2-T_1)}}\right)+\frac{\eta f_t^{T_2}}{Y_t}\left(1-\frac{\widetilde{\kappa}}{\widetilde{\gamma}}\right)\left(\frac{ e^{\widetilde{\kappa}(T_1-t)} }{e^{-\widetilde{\kappa}(T_2-T_1)}-e^{-\widetilde{\gamma}(T_2-T_1)}}\right).
\end{equation}
For example, take $\beta=1$ and $\eta=0$. Then, only the first terms in \eqref{u1eq} and \eqref{u2eq} remain, and $u_t^{(1)}$ is positive but $u_t^{(2)}$ is negative. This combination of futures contracts allows for direct exposure to the volatility index, without any exposure to the stochastic mean of the volatility.

\section{Concluding Remarks}\label{sec:conclusion}
We have studied a class of strategies  for  tracking an  index  and generating  desired exposure to a given  set of    factors in a general continuous-time diffusion framework.  Our analytical results provide the   tracking condition  and derivatives  trading strategies  under any given model. We have illustrated the path-behaviors of the trading strategies and the associated  tracking performances  under a number of well-known models, such as the Black-Scholes and  Heston models for equity tracking, and the CIR and CSQR models for VIX tracking.  This has practical implications to investors or fund managers  who seek to construct portfolios of derivatives for tracking purposes. Our results also shed light on  how tracking errors can arise  in a general financial market. 

There are many natural and useful extensions to this research. Our  illustrative  examples include equity indices and volatility indices, but there are many more asset classes, including    commodities, fixed income, and currencies. In all these asset classes, many  ETFs are designed to   track the same or similar indices. There is potential to  apply and extend our  methodology   to analyze  the price  dynamics and tracking performances of various futures/derivatives-based ETFs, and  derivatives portfolios in general. With options written on leveraged ETFs, investors can now use derivatives to generate   leveraged exposure. This calls for  consistent pricing of LETF options across leverage ratios (see \cite{Leung2014,Leung2014a}).  Moreover, our work shows that in some models it is impossible to  perfectly control exposure to the index and all factors. This should motivate the design of new trading strategies that minimize deviations from a  targeted exposure.  The insight on tracking errors  can     in principle be exploited for   statistical arbitrage, but it can also help  investors and regulators   better understand the risks associated with derivatives portfolios.

\newpage
\appendix

\section{Proofs}

\subsection{Derivation of  SDE \eqref{ckreturnSDE}}\label{App1}
By Ito's formula, the option price satisfies the SDE
\begin{equation}
\begin{aligned}
dc_t^{(k)}=\frac{\partial c^{(k)}}{\partial t}dt+dM_t^{\intercal}\nabla c^{(k)} +\frac{1}{2}dM_t^{\intercal} \nabla^2 c^{(k)}dM_t. \label{cSDE}
\end{aligned}
\end{equation}
Here, the gradient and Hessian are taken with respect to all market variables, $(S,Y^{(1)},...,Y^{(d)})$. We rewrite the last term in \eqref{cSDE} as follows:
 \begin{align}
dM_t^\intercal \nabla^2 c^{(k)}dM_t&=(\widetilde{\gamma}_t dt + \Sigma_t dB_t^\Q)^\intercal \nabla^2 c^{(k)}(\widetilde{\gamma}_t dt + \Sigma_t dB_t^\Q)\notag\\
&= \left(dB_t^\Q\right)^\intercal\Sigma_t^\intercal \nabla^2 c^{(k)}\Sigma_t dB_t^\Q\notag\\
&=\Tr\left[\left(dB_t^\Q\right)^\intercal\Sigma_t^\intercal \nabla^2 c^{(k)}\Sigma_t dB_t^\Q\right]\notag\\
&=\Tr\left[\Sigma_t^\intercal \nabla^2 c^{(k)}\Sigma_t dB_t^\Q\left(dB_t^\Q\right)^\intercal\right]\notag\\
&=\Tr\left[\Sigma_t^\intercal \nabla^2 c^{(k)}\Sigma_t \textbf{I}dt\right]=\Tr\left[\Sigma_t^\intercal \nabla^2 c^{(k)}\Sigma_t \right]dt.\label{lastterm}
\end{align}
 
Substituting \eqref{lastterm}  into \eqref{cSDE}, we obtain 
\begin{equation}\label{eq:c_ito_simple}
\begin{aligned}
dc_t^{(k)}=\left(\frac{\partial c^{(k)}}{\partial t}+\frac{1}{2}\Tr\left[\Sigma_t^\intercal \nabla^2 c^{(k)}\Sigma_t \right]\right)dt+\frac{\partial c^{(k)}}{\partial S}dS_t+\frac{\partial c^{(k)}}{\partial Y^{(1)}}dY_t^{(1)}+...+\frac{\partial c^{(k)}}{\partial Y^{(d)}}dY_t^{(d)}.
\end{aligned}
\end{equation}
On the other hand, the Feynman-Kac formula tells us that the derivative price satisfies
\begin{equation}
\begin{aligned}
\frac{\partial c^{(k)}}{\partial t}+\widetilde{\gamma}_t^{(0)}\frac{\partial c^{(k)}}{\partial S}+\widetilde{\gamma}_t^{(1)}\frac{\partial c^{(k)}}{\partial Y^{(1)}}+...+\widetilde{\gamma}_t^{(d)}\frac{\partial c^{(k)}}{\partial Y^{(d)}}+\frac{1}{2}\Tr\left[\Sigma_t^\intercal \nabla^2 c^{(k)}\Sigma_t \right]=rc^{(k)},
\end{aligned}
\end{equation}
with the terminal condition: $c^{(k)}(T_k,s,y^{(1)},...,y^{(d)})=h^{(k)}(s,y_1,...,y_d)$ for all vectors $(s,y_1,...,y_d)$ with strictly positive components. Using  \eqref{eq:c_ito_simple}, we have 
\begin{equation*}
\begin{aligned}
dc_t^{(k)}&=\left(rc_t^{(k)}-\widetilde{\gamma}_t^{(0)}\frac{\partial c^{(k)}}{\partial S}-\widetilde{\gamma}_t^{(1)}\frac{\partial c^{(k)}}{\partial Y^{(1)}}-...-\widetilde{\gamma}_t^{(d)}\frac{\partial c^{(k)}}{\partial Y^{(d)}}\right)dt\\
 &\quad\quad+\frac{\partial c^{(k)}}{\partial S}dS_t+\frac{\partial c^{(k)}}{\partial Y^{(1)}}dY_t^{(1)}+...+\frac{\partial c^{(k)}}{\partial Y^{(d)}}dY_t^{(d)}.\\
\end{aligned}
\end{equation*}
Dividing both sides by $c_t^{(k)}$ and using the definitions listed in   \eqref{eq:nonlin_elastDef}, we obtain  SDE \eqref{ckreturnSDE}.

\subsection{Validation of \eqref{eq:csqr_exist} when $\widetilde{\kappa}=\widetilde{\gamma}$}\label{App2}
We now demonstrate that, in the special case of  $\widetilde{\kappa}=\widetilde{\gamma}$ under  the CSQR model,   tracking strategies exist using  two  futures on $S$ of different maturities. The elasticity with respect to the index is the same as in \eqref{eq:csqr_spotE}; however, the elasticity with respect to the stochastic mean is
\begin{equation}\label{eq:app_meanE}
\begin{aligned}
H_t^{(k)}=\frac{Y_t}{f_t^{T_k}}\frac{\partial f_t^{T_k}}{\partial Y}=\frac{Y_t}{f_t^{T_k}}\widetilde{\gamma}e^{\widetilde{\kappa}t-\widetilde{\gamma}T_k}(T_k-t).
\end{aligned}
\end{equation}
Substituting \eqref{eq:app_meanE} into   \eqref{eq:csqr_exist}, we have
\begin{equation*}
\begin{aligned}
&\frac{S_t}{f_t^{T_1}}e^{-\widetilde{\gamma}(T_1-t)}\frac{Y_t}{f_t^{T_2}}\widetilde{\gamma}e^{\widetilde{\kappa}t-\widetilde{\gamma}T_2}(T_2-t)\neq\frac{S_t}{f_t^{T_2}}e^{-\widetilde{\gamma}(T_2-t)}\frac{Y_t}{f_t^{T_1}}\widetilde{\gamma}e^{\widetilde{\kappa}t-\widetilde{\gamma}T_1}(T_1-t)\\
\iff&e^{-\widetilde{\gamma}(T_1-t)}e^{\widetilde{\kappa}t-\widetilde{\gamma}T_2}(T_2-t)\neq e^{-\widetilde{\gamma}(T_2-t)}e^{\widetilde{\kappa}t-\widetilde{\gamma}T_1}(T_1-t) \iff T_1\neq T_2.\\
\end{aligned}
\end{equation*}
Under the assumption that $T_1<T_2$, the resulting system is solvable and yields   strategies that achieve the given exposures.

\subsection{Solution to \eqref{eq:line_sys_csqr} when $\widetilde{\kappa}=\widetilde{\gamma}$}\label{App3}
With the  elasticity with respect to the stochastic mean   given by   \eqref{eq:app_meanE} and $\widetilde{\kappa}=\widetilde{\gamma}$ under the CSQR model, the portfolio weights that solve system \eqref{eq:line_sys_csqr} are 
\begin{equation*}
\begin{aligned}
u_t^{(1)}=\frac{\beta f_t^{T_1}}{S_t}\left(\frac{ e^{\widetilde{\gamma}(T_1-t)}(T_2-t)}{T_2-T_1}\right)-\frac{\eta f_t^{T_1}}{Y_t}\left(\frac{e^{\widetilde{\gamma}(T_1-t)}}{\widetilde{\gamma}(T_2-T_1)}\right),\\
\end{aligned}
\end{equation*}
and
\begin{equation*}
\begin{aligned}
u_t^{(2)}=\frac{-\beta f_t^{T_2}}{S_t}\left(\frac{e^{\widetilde{\gamma}(T_2-t)}(T_1-t)}{T_2-T_1}\right)+\frac{\eta f_t^{T_2}}{Y_t}\left(\frac{e^{\widetilde{\gamma}(T_2-t)}}{\widetilde{\gamma}(T_2-T_1)}\right).\\
\end{aligned}
\end{equation*}
%\begin{equation*}
%\begin{aligned}
%u_t^{(1)}=\frac{\beta H_t^{(2)}-\eta G_t^{(2)}}{G_t^{(1)}H_t^{(2)}-G_t^{(2)}H_t^{(1)}}&=\frac{\beta \frac{Y_t}{f_t^{T_2}}\widetilde{\gamma}e^{-\widetilde{\gamma}(T_2-t)}(T_2-t)-\eta \frac{S_t}{f_t^{T_2}}e^{-\widetilde{\gamma}(T_2-t)}}{\frac{S_tY_t\widetilde{\gamma}}{f_t^{T_1}f_t^{T_2}}e^{-\widetilde{\gamma}(T_1-t)}e^{-\widetilde{\gamma}(T_2-t)}(T_2-T_1)}\\
%&=\frac{\beta f_t^{T_1} Y_t\widetilde{\gamma}(T_2-t)-\eta f_t^{T_1}S_t}{S_tY_t\widetilde{\gamma}e^{-\widetilde{\gamma}(T_1-t)}(T_2-T_1)}\\
%&=\frac{\beta f_t^{T_1}}{S_t}\left(\frac{ e^{\widetilde{\gamma}(T_1-t)}(T_2-t)}{T_2-T_1}\right)-\frac{\eta f_t^{T_1}}{Y_t}\left(\frac{e^{\widetilde{\gamma}(T_1-t)}}{\widetilde{\gamma}(T_2-T_1)}\right)\\
%\end{aligned}
%\end{equation*}
%and
%\begin{equation*}
%\begin{aligned}
%u_t^{(2)}=\frac{-\beta H_t^{(1)}+\eta G_t^{(1)}}{G_t^{(1)}H_t^{(2)}-G_t^{(2)}H_t^{(1)}}&=\frac{-\beta \frac{Y_t}{f_t^{T_1}}\widetilde{\gamma}e^{-\widetilde{\gamma}(T_1-t)}(T_1-t)+\eta \frac{S_t}{f_t^{T_1}}e^{-\widetilde{\gamma}(T_1-t)}}{\frac{S_tY_t\widetilde{\gamma}}{f_t^{T_1}f_t^{T_2}}e^{-\widetilde{\gamma}(T_1-t)}e^{-\widetilde{\gamma}(T_2-t)}(T_2-T_1)}\\
%&=\frac{-\beta Y_tf_t^{T_2}\widetilde{\gamma}(T_1-t)+\eta S_tf_t^{T_2}}{S_tY_t\widetilde{\gamma}e^{-\widetilde{\gamma}(T_2-t)}(T_2-T_1)}\\
%&=\frac{-\beta f_t^{T_2}}{S_t}\left(\frac{e^{\widetilde{\gamma}(T_2-t)}(T_1-t)}{T_2-T_1}\right)+\frac{\eta f_t^{T_2}}{Y_t}\left(\frac{e^{\widetilde{\gamma}(T_2-t)}}{\widetilde{\gamma}(T_2-T_1)}\right).\\
%\end{aligned}
%\end{equation*}
In particular,  take $\beta=1$ and $\eta=0$. Then, the portfolio weights simplify to 
\begin{equation*}\
\begin{aligned}
u_t^{(1)}=\frac{ f_t^{T_1}}{S_t}\left(\frac{ e^{\widetilde{\gamma}(T_1-t)}(T_2-t)}{T_2-T_1}\right), \quad \text{and} \quad
u_t^{(2)}=\frac{- f_t^{T_2}}{S_t}\left(\frac{e^{\widetilde{\gamma}(T_2-t)}(T_1-t)}{T_2-T_1}\right).\\
\end{aligned}
\end{equation*}
By inspection,  $u_t^{(1)}$ is positive  while $u_t^{(2)}$ is negative.

\subsection{Validation of \eqref{eq:csqr_exist} when $\widetilde{\kappa}\neq \widetilde{\gamma}$}\label{App4}
We substitute the elasticities from   \eqref{eq:csqr_spotE} into   \eqref{eq:csqr_exist}, and then check directly whether or not the following holds:
\begin{equation*}\label{eq:CIR_strat_det}
\begin{aligned}
&\frac{S_t}{f_t^{T_1}}e^{-\widetilde{\gamma}(T_1-t)}\frac{Y_t}{f_t^{T_2}}\frac{\widetilde{\gamma}}{\widetilde{\gamma}-\widetilde{\kappa}}\left(e^{-\widetilde{\kappa}(T_2-t)}-e^{-\widetilde{\gamma}(T_2-t)}\right)\neq\frac{S_t}{f_t^{T_2}}e^{-\widetilde{\gamma}(T_2-t)}\frac{Y_t}{f_t^{T_1}}\frac{\widetilde{\gamma}}{\widetilde{\gamma}-\widetilde{\kappa}}\left(e^{-\widetilde{\kappa}(T_1-t)}-e^{-\widetilde{\gamma}(T_1-t)}\right)\\
& \qquad\qquad \iff e^{-\widetilde{\gamma}T_1}\left(e^{-\widetilde{\kappa}(T_2-t)}-e^{-\widetilde{\gamma}(T_2-t)}\right)\neq e^{-\widetilde{\gamma}T_2}\left(e^{-\widetilde{\kappa}(T_1-t)}-e^{-\widetilde{\gamma}(T_1-t)}\right)
\iff T_1\neq T_2.
\end{aligned}
\end{equation*}
Since we assume $T_1<T_2$, the resulting system is solvable and allows for strategies to attain the required  exposures.

\begin{small} 
	\bibliographystyle{apa}
	\bibliography{mybib2_10222015}
\end{small}

\end{document}